\documentclass[a4paper,onecolumn,11pt,accepted=2022-08-24]{quantumarticle}
\pdfoutput=1
\usepackage[utf8]{inputenc}
\usepackage[english]{babel}
\usepackage[T1]{fontenc}
\usepackage{amsmath}
\usepackage{dsfont}
\usepackage{amssymb}
\usepackage{amsthm}
\usepackage{mathtools}
\usepackage{tikz}
\usetikzlibrary{calc,shapes}
\usepackage[numbers]{natbib}
\usepackage{hyperref}

\newtheorem{theorem}{Theorem}
\newtheorem{lemma}[theorem]{Lemma}

\newtheorem{cor}[theorem]{Corollary}
\theoremstyle{definition}

\newcommand{\abs}[1]{\lvert #1 \rvert}

\newcommand{\ip}[2]{\langle #1 , #2\rangle}
\newcommand{\bigip}[2]{\bigl\langle #1, #2 \bigr\rangle}

\newcommand{\bigceil}[1]{\bigl\lceil #1 \bigr\rceil}

\newcommand{\ket}[1]{\lvert\microspace #1 \microspace \rangle}
\newcommand{\bigket}[1]{\bigl\lvert\microspace #1 \microspace \bigr\rangle}

\newcommand{\bra}[1]{\langle\microspace #1 \microspace \rvert}
\newcommand{\bigbra}[1]{\bigl\langle\microspace #1 \microspace \bigr\rvert}

\newcommand\complex{\mathbb{C}}
\newcommand\real{\mathbb{R}}

\newcommand{\tinyspace}{\mspace{1mu}}
\newcommand{\tsp}{\mspace{1mu}}
\newcommand{\microspace}{\mspace{0.5mu}}

\newcommand{\op}[1]{\operatorname{#1}}
\newcommand{\tr}{\operatorname{Tr}}

\newcommand{\norm}[1]{\lVert\tinyspace #1 \tinyspace\rVert}
\newcommand{\bignorm}[1]{\bigl\lVert\tinyspace #1 \tinyspace\bigr\rVert}
\newcommand{\Bignorm}[1]{\Bigl\lVert\tinyspace #1 \tinyspace\Bigr\rVert}
\newcommand{\biggnorm}[1]{\biggl\lVert\tinyspace #1 \tinyspace\biggr\rVert}

\newcommand\I{\mathds{1}}

\newcommand{\setft}[1]{\mathrm{#1}}
\newcommand{\Density}{\setft{D}}
\newcommand{\Pos}{\setft{Pos}}

\newcommand{\Unitary}{\setft{U}}
\newcommand{\Herm}{\setft{Herm}}
\newcommand{\Lin}{\setft{L}}
\newcommand{\Channel}{\setft{C}}

\newenvironment{mylist}[1]{\begin{list}{}{
	\setlength{\leftmargin}{#1}
	\setlength{\rightmargin}{0mm}
	\setlength{\labelsep}{2mm}
	\setlength{\labelwidth}{8mm}
	\setlength{\itemsep}{0mm}}}
	{\end{list}}

\newcommand{\reg}[1]{\mathsf{#1}}

\newcommand\X{\mathcal{X}}
\newcommand\Y{\mathcal{Y}}

\newcommand\B{\mathcal{B}}
\newcommand\V{\mathcal{V}}

\newcommand\C{\mathcal{C}}

\renewcommand\S{\mathcal{S}}

\newcommand\K{\mathcal{K}}
                 
\definecolor{White}{rgb}{1,1,1}
\definecolor{Black}{rgb}{0,0,0}
\definecolor{LightGray}{rgb}{.81,.81,.81}
\colorlet{ChannelColor}{LightGray}
\colorlet{ChannelTextColor}{Black}
\colorlet{ReadoutColor}{White}

\begin{document}

\title{%
  Quantum game theory and the complexity of\newline
  approximating quantum Nash equilibria}

%

\author{John Bostanci}

\affil{%
  Computer Science Department,
  Columbia University
}


\author{John Watrous}

\affil{%
  Institute for Quantum Computing and School of Computer
  Science, University of Waterloo
}
\orcid{0000-0002-4263-9393}

\maketitle

\begin{abstract}
  This paper is concerned with complexity theoretic aspects of a general
  formulation of quantum game theory that models strategic interactions among
  rational agents that process and exchange quantum information.
  In particular, we prove that the computational problem of finding an
  approximate Nash equilibrium in a broad class of quantum games is, like
  the analogous problem for classical games, included in
  (and therefore complete for) the complexity class $\mathrm{PPAD}$.
  Our main technical contribution, which facilitates this inclusion, is an
  extension of prior methods in computational game theory to strategy spaces
  that are characterized by semidefinite programs.
\end{abstract}


\section{Introduction}

Game theory is a fascinating topic of study with connections to computer
science, economics, and the social sciences, among other subjects.
This paper focuses on complexity theoretic aspects of game theory within the
context of quantum information and computation.

Quantum game theory began with the work of David Meyer \cite{Meyer1999}
and Jens Eisert, Martin Wilkens, and Maciej Lewenstein
\cite{EisertWL1999} in 1999.\footnote{
  Some authors argue that the origins of quantum game theory go back further.
  Here, however, we are referring to the specific line of work that
  self-identifies as being concerned with a quantum information theoretic
  variant of game theory in the tradition of von~Neumann and Morgenstern
  \cite{vonNeumannM1953} and Nash \cite{Nash1950,Nash1950-thesis}, as opposed
  to quantum information and computation research that can be associated with
  game theory as a broad umbrella term.}
These works investigated games involving quantum information, highlighting
examples in which quantum players have advantages over classical players.
Many other examples of quantum games, primarily based on the frameworks
proposed by Meyer and Eisert, Wilkens, and Lewenstein, were subsequently
analyzed.
(See, for instance, the survey \cite{GuoZK2008} for summaries and references.)

Aspects of this line of work have been criticized for multiple reasons.
A common point of criticism of many (but certainly not all) quantum game theory
papers is their poorly motivated notion of classical behavior.
In particular, classical players in quantum game theory papers are often
limited to \emph{coherent} permutations of standard basis states, or
similarly restricted classes of unitary operations, while quantum players have
access to a less restricted set of unitary operations, possibly all of them.
This notion of classicality, which is a key ingredient in the original examples
of Meyer and Eisert, Wilkens, and Lewenstein, essentially invites exploitation
by quantum players.
A more standard interpretation of classical behavior in quantum information
theory assumes the complete decoherence of any quantum system a classical
player manipulates.


Another point of criticism, raised by van Enk and Pike \cite{vanEnkP2002}, is
that comparing quantum games with their classical namesakes within the specific
frameworks typically adopted by quantum game theory papers is akin to comparing
apples with oranges.
Although one may argue that these games offer faithful representations of
classical games when players' actions are restricted to permutations of
standard basis states, their quantum reformulations are, simply  put,
different games.
It is therefore not surprising that less restricted quantum players may find
advantages, leading to new Nash equilibria.

However, although it was not their primary focus, Meyer and 
Eisert, Wilkens, and Lewenstein did both clearly suggest more general
definitions of quantum games in which a wide range of interactions could be
considered, including ones in which the criticisms just raised no longer have
relevance.
In particular, Meyer mentions a convex form of his model of quantum games,
in which classical players could be modeled by completely decohered operations.
And, Eisert, Wilkens, and Lewenstein, in a footnote of their paper, describe a
model in which players' actions correspond not just to unitary operations, but
to arbitrary quantum channels (as modeled by completely positive and trace
preserving linear maps).
In either case, more general strategic interactions may be considered, and
one need not restrict their attention to analogues of classical games or in
identifying a ``quantum advantage.''

For example, quantum interactive proof systems of various sorts, as well as
many quantum cryptographic scenarios and primitives, can be viewed as quantum
games.
Another example is quantum communication, which can be modeled as a game
in which one player attempts to transmit a quantum state to another, while a
third player representing an adversarial noise model attempts to disrupt the
transmission.
We do not offer any specific suggestions in this paper, but it is not
unreasonable to imagine that quantum games having social or economic
applications could be discovered.

We will now summarize the definition of quantum games we adopt, beginning with
the comparatively simple non-interactive setting and then moving on to the more
general interactive setting.
For the sake of simplicity and exposition in this introduction, we will
restrict our attention to games in which there are just two players: Alice and
Bob.
The definitions are easily extended to any finite number of players, as is
done in the main text.

\subsubsection*{Non-interactive quantum games}

In a (two-player) \emph{non-interactive quantum game}, the players Alice and
Bob each hold a quantum system, represented by a register of a predetermined
size: Alice holds $\reg{X}$ and Bob holds $\reg{Y}$.
They must each \emph{independently} prepare in the register they hold a
quantum state: Alice prepares a quantum state represented by a density
operator $\rho$ and Bob prepares a state represented by $\sigma$.
Just like in the standard \emph{non-cooperative} setting of classical game
theory, Alice and Bob are assumed to be unable to correlate their state
preparations with one another.\footnote{%
  It is interesting to consider meaningful ways in which this assumption may be
  relaxed or dropped, but the simplest and most direct quantum extension
  of classical game theory beings with this assumption of independence.}
The registers $\reg{X}$ and $\reg{Y}$ are sent to a referee, who performs a
joint measurement on the pair $(\reg{X},\reg{Y})$.
Here, when we refer to a measurement, we mean a general quantum
measurement, often called a POVM (positive operator valued measure), having any
finite and non-empty set of measurement outcomes.
The outcome of the measurement is assumed to determine a real number payoff for
each player.
We note explicitly that Jinshan Wu \cite{Wu2004a,Wu2004b} has proposed and
analyzed an equivalent definition of non-interactive quantum games to this
one.

In order to formally describe a non-interactive quantum game, one must
specify the referee's measurement together with the payoff functions for each
player.
As will be explained later, when it suffices to specify each player's
\emph{expected} payoff, given any choice of states the players may select,
the referee may be described by a collection of Hermitian matrices, one for
each player.

One may observe that the standard notion of a classical game in normal form is
easily represented within this framework by defining the referee so that it
first measures the registers $\reg{X}$ and $\reg{Y}$ with respect to the
standard bases of the associated spaces, and then assigns payoffs in a
completely classical manner.

A non-interactive quantum game can, up to a discretization, also be viewed as a
classical game, where the players send the referee classical descriptions of
their chosen density operators and the referee performs the required
calculation to determine their payoffs, but the normal form description of this
new classical game will, naturally, be exponentially larger than the
description of the original quantum game.

\subsubsection*{Interactive quantum games}

Quantum games in which players can process and exchange quantum
information with a referee over the course of multiple rounds of interaction
may also be considered.

For example, the referee could prepare registers $\reg{X}$ and~$\reg{Y}$ in a
joint quantum state, send $\reg{X}$ to Alice and $\reg{Y}$ to Bob, allowing
them to transform these registers as they choose, and then measure the pair
$(\reg{X},\reg{Y})$ upon receiving them back from Alice and Bob.
In such a game, Alice and Bob therefore each play a quantum channel, with their
payoffs again being determined by the outcome of the referee's measurement.
The framework of Eisert, Wilkens, and Lewenstein falls into this category,
provided that the players are permitted to play channels and not just unitary
operations.

Zhang~\cite{Zhang2012} introduced and studied a related model, where
the referee distributes a quantum state to the players, who then effectively
choose local measurements as their stratgies.
Through this model Zhang identified interesting aspects of so-called
\emph{correlated equilibria} in quantum games.

More generally, an \emph{interactive quantum game} may involve an interaction
between the referee and the players over the course of multiple rounds, as
suggested by Figure~\ref{figure:interactive-game}.
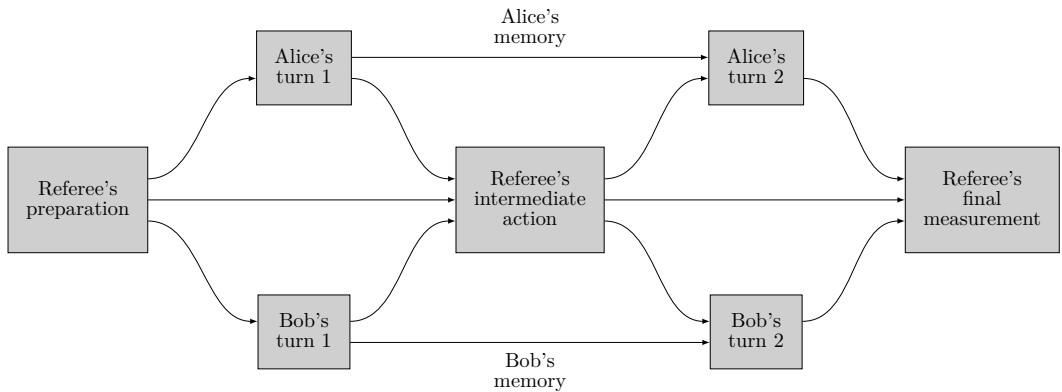
\begin{figure}
  \begin{center}
    \scalebox{0.7}{%
      \begin{minipage}{19.9cm}
        \begin{tikzpicture}[
            scale = 0.5,
            turn/.style={draw, minimum height=14mm, minimum width=10mm,
              fill = ChannelColor, text=ChannelTextColor},
            invisibleturn/.style={minimum height=14mm, minimum width=10mm},
            bigturn/.style={draw, minimum height=20mm, minimum width=20mm,
              fill = ChannelColor, text=ChannelTextColor},
            measure/.style={draw, minimum width=7mm, minimum height=7mm,
              fill = ChannelColor},
            >=latex]
          
          \node (A1) at (-8.5,5) [turn] {%
            \begin{tabular}{c}
              Alice's\\[-1mm]turn 1
            \end{tabular}
          };

          \node (A2) at (8.5,5) [turn] {%
            \begin{tabular}{c}
              Alice's\\[-1mm]turn 2
            \end{tabular}
          };      

          \node (B1) at (-8.5,-5) [turn] {%
            \begin{tabular}{c}
              Bob's\\[-1mm]turn 1
            \end{tabular}
          };

          \node (B2) at (8.5,-5) [turn] {%
            \begin{tabular}{c}
              Bob's\\[-1mm]turn 2
            \end{tabular}
          };

          \node (R0) at (-17,0) [bigturn] {%
            \begin{tabular}{c}
              Referee's\\[-1mm]
              preparation
            \end{tabular}
          };

          \node (R1) at (0,0) [bigturn] {%
            \begin{tabular}{c}
              Referee's\\[-1mm]
              intermediate\\[-1mm]
              action
            \end{tabular}
          };

          \node (R2) at (17,0) [bigturn] {%
            \begin{tabular}{c}
              Referee's\\[-1mm]
              final\\[-1mm]
              measurement
            \end{tabular}
          };
          
          \draw[->] ([yshift=4mm]A1.east) -- ([yshift=4mm]A2.west)
          node [above, midway] {%
            \begin{tabular}{c}
              Alice's\\[-1mm]memory
          \end{tabular}};
          
          \draw[->] ([yshift=-4mm]A1.east) .. controls +(right:20mm) and 
          +(left:20mm) .. ([yshift=8mm]R1.west) node [right, pos=0.4]
          {};
          
          \draw[->] ([yshift=-4mm]A2.east) .. controls +(right:20mm) and 
          +(left:20mm) .. ([yshift=8mm]R2.west) node [right, pos=0.4] 
          {};
          
          \draw[->] ([yshift=8mm]R0.east) .. controls +(right:20mm) and 
          +(left:20mm) .. ([yshift=-4mm]A1.west) node [left, pos=0.6]
          {};
          
          \draw[->] ([yshift=8mm]R1.east) .. controls +(right:20mm) and 
          +(left:20mm) .. ([yshift=-4mm]A2.west) node [left, pos=0.6]
          {};

          \draw[->] (R0.east) -- (R1.west) node [below, midway] {};
          
          \draw[->] (R1.east) -- (R2.west) node [below, midway] {};
          
          \draw[->] ([yshift=-4mm]B1.east) -- ([yshift=-4mm]B2.west)
          node [below, midway] {%
            \begin{tabular}{c}
              Bob's\\[-1mm]memory
          \end{tabular}};
          
          \draw[->] ([yshift=4mm]B1.east) .. controls +(right:20mm) and 
          +(left:20mm) .. ([yshift=-8mm]R1.west) node [right, pos=0.4] {};
          
          \draw[->] ([yshift=4mm]B2.east) .. controls +(right:20mm) and 
          +(left:20mm) .. ([yshift=-8mm]R2.west) node [right, pos=0.4] {};
          
          \draw[->] ([yshift=-8mm]R0.east) .. controls +(right:20mm) and 
          +(left:20mm) .. ([yshift=4mm]B1.west) node [left, pos=0.6] {};
          
          \draw[->] ([yshift=-8mm]R1.east) .. controls +(right:20mm) and 
          +(left:20mm) .. ([yshift=4mm]B2.west) node [left, pos=0.6] {};
          
        \end{tikzpicture}
    \end{minipage}}
  \end{center}
  \caption{An illustration of a game between Alice and Bob, run by a referee,
    in which Alice and Bob each receive and transmit quantum information twice,
    potentially keeping a quantum memory between the two turns.
    Each arrow represents a quantum register, which could be of any fixed
    size (including the possibility of trivial registers, which are equivalent
    to nothing being transmitted).
    Games involving any finite number of rounds of interaction may be
    considered.}
  \label{figure:interactive-game}
\end{figure}
In this setting it is natural to assume that Alice and Bob each have their own
private quantum memory, which they utilize if it is to their advantage.

The actions of the players in such a game may be represented through the
framework alternatively known as the
\emph{quantum strategies} framework \cite{GutoskiW2007} and the
\emph{quantum combs} framework \cite{ChiribellaDP2008,ChiribellaDP2009}.
This framework, which will be described in greater detail in the next
section, allows the actions of any one player over the course of multiple
rounds of interaction, accounting for the possibility of a quantum memory, to
be faithfully represented by a single positive semidefinite matrix that
satisfies a finite collection of affine linear constraints.
Thus, the sets of strategies available to the players are convex and compact,
and one may efficiently optimize real-valued linear functions defined on these
sets through the paradigm of semidefinite programming.


Similar to non-interactive quantum games, interactive quantum games are
formally expressed by specifying the referee's actions, including state
preparations, channels, and measurements, along with payoff functions of the
possible measurement outcomes corresponding to each player.
Once again, when it is sufficient to describe the expected payoff for each
player, given a specification of their strategies, a referee in an interactive
quantum game may be specified by a list of Hermitian matrices, one for each
player, as will be explained.


\subsection*{Our results and contributions}

We prove, as our main technical result, that the problem of computing an
approximate Nash equilibrium in any interactive game of the form described
above, given an explicit matrix representation of the referee, is contained in
the complexity class PPAD.
As this problem includes non-interactive classical games as a special case, it
follows that this problem is complete for PPAD
\cite{DaskalakisGP2009,ChenDT2009}.

There is a sense in which this result is not unexpected; prior work on the
complexity of computing approximate Nash equilibria, and more generally on the
complexity of computing fixed points of different classes of continuous maps,
suggests that approximations of Nash equilibria in a wide variety of games
should be contained in PPAD
\cite{Papadimitriou1994,DaskalakisGP2009,ChenDT2009,EtessamiY2010}.
The principal challenge that arises in the setting of interactive quantum games
is that, although one may efficiently optimize over individual player's
strategies through semidefinite programming, closed form expressions of
these optimizations are not known to exist.

To confront this challenge, we consider a fairly general convex form of
Nash's notion of a \emph{gain function}---and then we fight fire with fire, so
to speak, using semidefinite programming to approximate continuous functions
that arise through this formulation.
Possibly our methodology for handling this issue will be of independent
interest.

Although it is not an essential aspect of our proof, we also make use of the
elegant notion of a \emph{discrete Wigner representation} of a quantum state,
which is convenient within the proof.
Although discrete Wigner representations have been investigated in the theory
of quantum information
(see, for instance, \cite{GibbonsHW2004} and \cite{Gross2006}), we are not
aware that they have been used previously in quantum complexity theory, and we
feel that they offer a convenient tool that might be useful in other contexts.

Beyond this main technical result, we hope that this paper may serve as a
suggestion that quantum game theory is worthy of a second look.
We believe that the general definition of quantum games we have reiterated is
well motivated by the theory of quantum information, and can provide a basic
foundation through which quantum game theory and its complexity theoretic
aspects may be investigated.
In the conclusion of this paper we mention several open problems and
research directions concerning quantum game theory that may be of interest.


\section{Technical preliminaries}

This section summarizes technical concepts required later in the paper.
The first three subsection that follow discuss aspects of the
\emph{barycentric subdivision} of a simplex, the
\emph{discrete Wigner representation} of quantum states, and
the \emph{quantum strategies/combs framework}, respectively.
In the last subsection we present the standard definition of the complexity
class $\mathrm{PPAD}$ and reference a theorem of Etessami and
Yannakakis~\cite{EtessamiY2010} that establishes the containment of a
certain computational fixed-point problem in $\mathrm{PPAD}$, to which we
will reduce the problem of finding an approximate Nash equilibrium of a
quantum game.

It will be assumed throughout the paper that the reader is familiar with basic
notions of computational complexity \cite{AroraB2009} and quantum information
\cite{NielsenC2000,Wilde2017,Watrous2018}.
Hereafter we will take $\Sigma = \{0,1\}$ to denote the binary alphabet.


\subsection{Barycentric subdivision of a simplex}

Suppose that a positive integer $n$ is given, and consider a simplex in an
$n$-dimensional space having vertices $\{u_1,\ldots,u_n\}$.
The \emph{barycentric subdivision} of such a simplex is a division of it into
$n!$ simplices in the following way.

First, with each nonempty subset $A\subseteq\{1,\ldots,n\}$ we define a point
\begin{equation}
  v_A = \frac{1}{\abs{A}}\sum_{k\in A} u_k,
\end{equation}
which is the uniform convex combination, or \emph{barycenter}, of the vertices
of the original simplex labeled by elements of $A$.
For example, $v_{\{k\}} = u_k$ for each $k\in\{1,\ldots,n\}$, while
$v_{\{1,\ldots,n\}}$ is the true barycenter of the original simplex.
Figure~\ref{fig:barycentric-1} illustrates the barycentric subdivision for a
simplex when $n=3$.
\begin{figure}[t]
  \begin{center}
  \begin{tikzpicture}[scale=4,>=latex]]

    \node (1) at (0,1) {};
    \node (2) at (-0.8660,-0.5) {};
    \node (3) at (0.8660,-0.5) {};

    \node (12) at (barycentric cs:1=1,2=1,3=0) {};
    \node (13) at (barycentric cs:1=1,2=0,3=1) {};
    \node (23) at (barycentric cs:1=0,2=1,3=1) {};

    \node (123) at (barycentric cs:1=1,2=1,3=1) {};

    \fill[black!10] (1.center) -- (12.center) -- (123.center) -- (1.center);

    \node at ([xshift=-1,yshift=-2]barycentric cs:1=1,12=1,123=1) {%
      \footnotesize $(1,2,3)$};

    \node at ([xshift=1,yshift=-2]barycentric cs:1=1,13=1,123=1) {%
      \footnotesize $(1,3,2)$};

    \node at ([xshift=1,yshift=1]barycentric cs:2=1,12=1,123=1) {%
      \footnotesize $(2,1,3)$};

    \node at ([xshift=-1,yshift=1]barycentric cs:3=1,13=1,123=1) {%
      \footnotesize $(3,1,2)$};

    \node at ([xshift=1,yshift=0]barycentric cs:2=1,23=1,123=1) {%
      \footnotesize $(2,3,1)$};

    \node at ([xshift=-1,yshift=0]barycentric cs:3=1,23=1,123=1) {%
      \footnotesize $(3,2,1)$};
    
    \draw[fill] (1.center) circle (0.5pt);
    \draw[fill] (2.center) circle (0.5pt);
    \draw[fill] (3.center) circle (0.5pt);
    \draw[fill] (12.center) circle (0.5pt);
    \draw[fill] (23.center) circle (0.5pt);
    \draw[fill] (13.center) circle (0.5pt);
    \draw[fill] (123.center) circle (0.5pt);
    
    \draw (1.center) -- (12.center);
    \draw (1.center) -- (13.center);
    \draw (1.center) -- (123.center);
    \draw (2.center) -- (12.center);
    \draw (2.center) -- (23.center);
    \draw (2.center) -- (123.center);
    \draw (3.center) -- (13.center);
    \draw (3.center) -- (23.center);
    \draw (3.center) -- (123.center);
    
    \draw (12.center) -- (123.center);
    \draw (13.center) -- (123.center);
    \draw (23.center) -- (123.center);
    
    \node at ([yshift=5pt]1.center) {$v_{\{1\}} = u_1$};
    \node at ([xshift=-4pt,yshift=-4pt]2.center) {$v_{\{2\}} = u_2$};
    \node at ([xshift=4pt,yshift=-4pt]3.center) {$v_{\{3\}} = u_3$};
    \node at ([xshift=-6pt,yshift=3pt]12.center) {$v_{\{1,2\}}$};
    \node at ([xshift=6pt,yshift=3pt]13.center) {$v_{\{1,3\}}$};
    \node at ([yshift=-4pt]23.center) {$v_{\{2,3\}}$};
    
    \node (barycenterlabel) at (1.25,0.2) {$v_{\{1,2,3\}}$};
    \draw[->] (barycenterlabel.west) -- (123.east);
    \node (counterbalance) at (-1.25,0) {\phantom{$\{1,2,3\}$}};
    
  \end{tikzpicture}
  \end{center}
  \caption{The barycentric subdivision of a simplex with three vertices
    $u_1$, $u_2$, and $u_3$.
    The shaded region indicates one of the $3!=6$ simplices formed by this
    subdivision: it is the one with vertices
    $\{v_{\{1\}}, v_{\{1,2\}}, v_{\{1,2,3\}}\}$, which is naturally
    identified with the identity permutation $\pi = (1,2,3)$.
  }
  \label{fig:barycentric-1}
\end{figure}
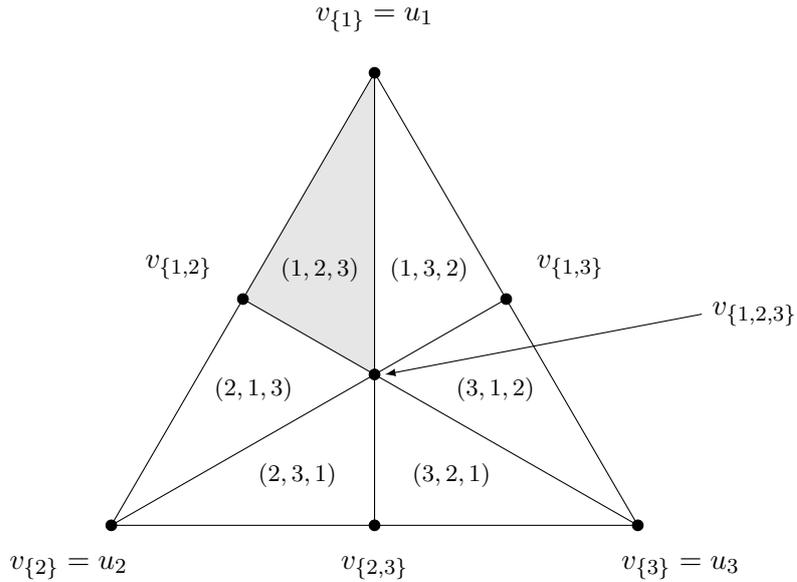

Next, by drawing an edge between $v_A$ and $v_B$ if and only if $A\subset B$
or $B\subset A$ (proper containments), we obtain a division of the original
simplex into $n!$ new simplices, one for each possible chain
\begin{equation}
  \label{eq:subset-chain}
  A_1 \subset \cdots \subset A_n
\end{equation}
of subsets of $\{1,\ldots,n\}$.
There are $n!$ such chains and they may be placed in correspondence with
the symmetric group $S_n$.
To be precise, for any fixed ordering $(k_1,\ldots,k_n)$ of the set
$\{1,\ldots,n\}$, we associate the chain \eqref{eq:subset-chain} with the
permutation $\pi\in S_n$ satisfying
\begin{equation}
  A_j = \{k_{\pi(1)},\ldots,k_{\pi(j)}\}
\end{equation}
for every $j\in \{1,\ldots, n\}$.
Thus, the simplices in the subdivision are identified with elements of $S_n$.

The barycentric subdivision may naturally be applied iteratively within
the simplices constructed by the subdivision.
For example, Figure~\ref{fig:barycentric-2} illustrates the barycentric
subdivision of just the shaded simplex illustrated in
Figure~\ref{fig:barycentric-1}.
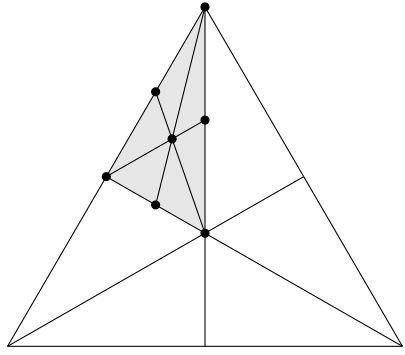
\begin{figure}[t]
\begin{center}
  \begin{tikzpicture}[scale=3,>=latex]]

    \node (1) at (0,1) {};
    \node (2) at (-0.8660,-0.5) {};
    \node (3) at (0.8660,-0.5) {};

    \node (12) at (barycentric cs:1=1,2=1,3=0) {};
    \node (13) at (barycentric cs:1=1,2=0,3=1) {};
    \node (23) at (barycentric cs:1=0,2=1,3=1) {};
    \node (123) at (barycentric cs:1=1,2=1,3=1) {};
    \node (1_12) at (barycentric cs:1=1,12=1) {};
    \node (12_123) at (barycentric cs:12=1,123=1) {};
    \node (1_123) at (barycentric cs:1=1,123=1) {};
    \node (1_12_123) at (barycentric cs:1=1,12=1,123=1) {};

    \fill[black!10] (1.center) -- (12.center) -- (123.center) -- (1.center);

    \draw[fill] (1_12.center) circle (0.5pt);
    \draw[fill] (12_123.center) circle (0.5pt);
    \draw[fill] (1_123.center) circle (0.5pt);
    \draw[fill] (1_12_123.center) circle (0.5pt);
    \draw[fill] (1.center) circle (0.5pt);
    \draw[fill] (12.center) circle (0.5pt);
    \draw[fill] (123.center) circle (0.5pt);

    \draw (1.center) -- (12.center);
    \draw (1.center) -- (13.center);
    \draw (1.center) -- (123.center);
    \draw (2.center) -- (12.center);
    \draw (2.center) -- (23.center);
    \draw (2.center) -- (123.center);
    \draw (3.center) -- (13.center);
    \draw (3.center) -- (23.center);
    \draw (3.center) -- (123.center);
    \draw (12.center) -- (123.center);
    \draw (13.center) -- (123.center);
    \draw (23.center) -- (123.center);
    \draw (1.center) -- (1_12_123.center);
    \draw (1_12.center) -- (1_12_123.center);
    \draw (1_123.center) -- (1_12_123.center);
    \draw (12.center) -- (1_12_123.center);
    \draw (123.center) -- (1_12_123.center);
    \draw (12_123.center) -- (1_12_123.center);
    
    
    
       
       
    
    
    
   
  \end{tikzpicture}
\end{center}
\caption{The barycentric subdivision applied to the simplex shaded gray in
  Figure~\ref{fig:barycentric-1}.}
\label{fig:barycentric-2}
\end{figure}
Hereafter we shall assume that the initial simplex is the standard
simplex $\Delta_n$, so that $u_1,\ldots,u_n$ are elementary unit vectors
(or, equivalently, standard basis vectors).
With this assumption in place, we define a sequence of finite subsets
of the standard unit simplex
\begin{equation}
  \B_n^0 \subset \B_n^1 \subset \B_n^2 \subset \cdots \subset \Delta_n;
\end{equation}
$\B_n^0$ contains the $n$ vertices of the simplex $\Delta_n$,
$\B_n^1$ contains the $2^n-1$ points corresponding to the nonempty subsets
of $\{1,\ldots,n\}$ when the barycentric subdivision has been applied a single
time, and in general $\B_n^r$ denotes the set of vertices after the $r$-th
level subdivision has been performed.

For any function $f:\Delta_n\rightarrow\Delta_n$, and any nonnegative integer
$r$, the \emph{$r$-th level barycentric approximation to $f$} is the function
$g_r:\Delta_n\rightarrow\Delta_n$ defined in the following way.
Each point $v\in\Delta_{n}$ may be expressed uniquely as a convex combination
\begin{equation}
  v = \lambda_1 v_1 + \cdots + \lambda_m v_m
\end{equation}
of distinct vertices $v_1,\ldots,v_m \in \B_n^r$, all contained in the same
$r$-th level simplex (which therefore implies $m\leq n$).
One then defines
\begin{equation}
  g_r(v) = \lambda_1 f(v_1) + \cdots + \lambda_m f(v_m).
\end{equation}

Various computations involving the $r$-th level barycentric subdivision may be
performed efficiently and exactly through rational number computations for $r$
being polynomial in the size of the problem being considered, but we will not
have a need to explicitly refer to these computations.
We will make use of the well known fact that any two points $u$ and $v$
contained in the same simplex constructed at the $r$-th level of the
barycentric subdivision must satisfy
\begin{equation}
  \norm{u - v}_2 \leq \Bigl(1 - \frac{1}{n+1}\Bigr)^r,
\end{equation}
with the norm being the Euclidean norm on $\real^n$.
A proof of this fact may be found in many texts on algebraic topology,
including in \cite{Bredon1993} where it appears as Lemma~17.3.


\subsection{Discrete Wigner representation}
\label{sec:discrete-Wigner}

Throughout this subsection we will take $n$ to be an odd positive integer,
and let us rename the elements of the standard basis
$\{\ket{1},\ldots,\ket{n}\}$ as $\{\ket{a}\,:\,a\in\mathbb{Z}_n\}$ by
taking each index modulo $n$.
In general, we shall interpret expressions inside of bras and kets as referring
to modulo $n$ arithmetic.

It is helpful to begin with the definition of the
\emph{discrete Weyl operators}.
First define
\begin{equation}
  X = \sum_{a\in\mathbb{Z}_n} \ket{a + 1}\bra{a}
  \quad\text{and}\quad
  Z = \sum_{a\in\mathbb{Z}_n} \omega_n^a \tsp\ket{a}\bra{a}
\end{equation}
for $\omega_n = \exp(2\pi i/n)$ denoting the first principal $n$-th root of
unity.
The discrete Weyl operators
\begin{equation}
  \{W_{a,b}\,:\,a,b\in\mathbb{Z}_n\}\subset\Unitary(\complex^n),
\end{equation}
as we will define them, are then given by
\begin{equation}
  W_{a,b} = X^a Z^b
\end{equation}
for every $a,b\in\mathbb{Z}_n$.
The discrete Weyl operators form an orthogonal basis for the vector
space~$\Lin(\complex^n)$.

Next, define an operator $T\in\Lin(\complex^n)$ as
\begin{equation}
  T = \sum_{a\in\mathbb{Z}_n}\ket{-a}\bra{a}.
\end{equation}
For example, in dimension $n=5$ this operator may be expressed in matrix form
as
\begin{equation}
  T =
  \begin{pmatrix}
    0 & 0 & 0 & 1 & 0\\
    0 & 0 & 1 & 0 & 0\\
    0 & 1 & 0 & 0 & 0\\
    1 & 0 & 0 & 0 & 0\\
    0 & 0 & 0 & 0 & 1
  \end{pmatrix}.
\end{equation}
We then define
\begin{equation}
  V_{a,b} = W_{a,b} T W_{a,b}^{\ast}
\end{equation}
for every $a,b\in\mathbb{Z}_n$, and consider the collection
\begin{equation}
  \label{eq:Wigner-operators}
  \bigl\{V_{a,b} \,:\, a,b\in\mathbb{Z}_n\bigr\}.
\end{equation}
One may observe that $T$ is unitary, Hermitian, and, by the
assumption that $n$ is odd, has unit trace, and therefore the same is true
for every operator in the collection \eqref{eq:Wigner-operators}.

Next let us verify that the collection \eqref{eq:Wigner-operators} is
orthogonal, with respect to the usual (Hilbert--Schmidt) inner product on
operators.
For any choice of $a,b,c,d\in\mathbb{Z}_n$, one may verify directly that
\begin{equation}
  \bigip{V_{a,b}}{V_{c,d}}
  = \bigip{T}{V_{c-a,d-b}}.
\end{equation}
Also observe the following expression for the diagonal entries of the operator
$T V_{a,b}$:
\begin{equation}
  \bigbra{c}TV_{a,b}\bigket{c}
  = \bigbra{c}TW_{a,b}T W_{a,b}^{\ast}\bigket{c}
  = \begin{cases}
    0 & a \not= 0\\
    \omega_n^{-2bc} & a = 0.
  \end{cases}
\end{equation}

\noindent
Noting the expression
\begin{equation}
  \sum_{c\in\mathbb{Z}_n}
  \omega_n^{-2bc}
  = \begin{cases}
    n & b=0\\
    0 & b\not=0,
  \end{cases}
\end{equation}
where again we have used the assumption that $n$ is odd, we conclude that
\begin{equation}
  \bigip{T}{V_{a,b}} =
  \begin{cases}
    n & (a,b) = (0,0)\\
    0 & (a,b) \not= (0,0).
  \end{cases}
\end{equation}
The collection \eqref{eq:Wigner-operators} is therefore orthogonal.

At this point we have no further need to refer to modulo $n$ arithmetic, so let
us assume that the elements of the collection
\eqref{eq:Wigner-operators} have been renamed as $\{V_1,\ldots,V_{n^2}\}$,
with respect to any sensible way of doing this.
The key property of this collection is that each $V_k$ is unitary, Hermitian,
and has trace equal to 1, and that the collection is orthogonal.

Finally, define an affine linear map of the form
$\psi:\Herm(\complex^n) \rightarrow \real^{n^2}$ as
\begin{equation}
  \psi(H) = \frac{1}{n(n+1)} \sum_{k=1}^{n^2} (\ip{V_k}{H} + 1)\, \ket{k}
\end{equation}
for every $H\in\Herm(\complex^n)$.
The inverse of this mapping is given by
\begin{equation}
  \psi^{-1}(v) = (n+1) \sum_{k = 1}^{n^2} v_k V_k - \I_n
\end{equation}
for every $v\in\real^{n^2}$.
This mapping defines a \emph{discrete Wigner representation} of quantum states;
each density operator $\rho\in\Density(\complex^n)$ is represented by the
vector
\begin{equation}
  v = \psi(\rho) \in \Delta_{n^2}.
\end{equation}
The inclusion of the vector $v = \psi(\rho)$ in the unit simplex follows
from two observations, the first being that
$\tr(\psi^{-1}(v)) = 1$ if and only if $v_1 + \cdots + v_{n^2} = 1$, and the
second being that $\ip{V_k}{\rho}\in[-1,1]$ by virtue of the fact that $V_k$ is
unitary and Hermitian, for each $k\in\{1,\ldots,n^2\}$.

It should be noted that, except in the trivial case $n = 1$, the
inclusion $\psi(\Density(\complex^n)) \subset \Delta_{n^2}$ is proper;
only a subset of the vectors in the standard simplex represent a valid
density operator, others represent unit-trace Hermitian operators having
negative eigenvalues.

\subsection{The quantum strategies framework}
\label{sec:strategies}

We now summarize aspects of the quantum strategies/combs framework
\cite{GutoskiW2007,ChiribellaDP2008,ChiribellaDP2009}, hereafter be referred to
as the \emph{quantum strategies framework} in this paper, that are required for
our main result.
This framework provides a convenient way of describing and characterizing the
actions of agents that interact and exchange quantum information with one
another over the course of multiple rounds.

\begin{figure}
  \begin{center}
    \footnotesize
    \begin{tikzpicture}[scale=0.35, 
        turn/.style={draw, minimum height=14mm, minimum width=10mm,
          fill = ChannelColor, text=ChannelTextColor},
        measure/.style={draw, minimum width=7mm, minimum height=7mm,
          fill = ChannelColor},
        invisible/.style={minimum height=0mm, minimum width=10mm},
        >=latex]
      
      \node (V1) at (-8,4) [turn] {$\Phi_1$};
      \node (V2) at (2,4) [turn] {$\Phi_2$};
      \node (V3) at (12,4) [turn] {$\Phi_3$};
           
      \node (P0) at (-13,-2) [invisible] {};
      \node (P1) at (-3,-2) [invisible] {};
      \node (P2) at (7,-2) [invisible] {};
      \node (P3) at (17,-2) [invisible] {};

      \draw[->] ([yshift=4mm]V1.east) -- ([yshift=4mm]V2.west)
      node [above, midway] {$\reg{Z}_1$};
      
      \draw[->] ([yshift=4mm]V2.east) -- ([yshift=4mm]V3.west)
      node [above, midway] {$\reg{Z}_2$};
      
      
      \draw[->] ([yshift=-4mm]V1.east) .. controls +(right:20mm) and 
      +(left:20mm) .. ([yshift=4mm]P1.west) node [right, pos=0.5] {$\reg{Y}_1$};
      
      \draw[->] ([yshift=-4mm]V2.east) .. controls +(right:20mm) and 
      +(left:20mm) .. ([yshift=4mm]P2.west) node [right, pos=0.5] {$\reg{Y}_2$};
      
      \draw[->] ([yshift=-4mm]V3.east) .. controls +(right:20mm) and 
      +(left:20mm) .. ([yshift=4mm]P3.west) node [right, pos=0.5] {$\reg{Y}_3$};
      
      \draw[->] ([yshift=4mm]P0.east) .. controls +(right:20mm) and 
      +(left:20mm) .. ([yshift=-4mm]V1.west) node [left, pos=0.5] {$\reg{X}_1$};
      
      \draw[->] ([yshift=4mm]P1.east) .. controls +(right:20mm) and 
      +(left:20mm) .. ([yshift=-4mm]V2.west) node [left, pos=0.5] {$\reg{X}_2$};
      
      \draw[->] ([yshift=4mm]P2.east) .. controls +(right:20mm) and 
      +(left:20mm) .. ([yshift=-4mm]V3.west) node [left, pos=0.5] {$\reg{X}_3$};
            
    \end{tikzpicture}
  \end{center}
  \caption{
    The actions of an agent, or a \emph{strategy}, in a three-round interaction
    may be described by a network of three channels.
    Time goes from left to right: first the register $\reg{X}_1$ is received
    and fed into the channel $\Phi_1$, which produces $\reg{Y}_1$ and
    $\reg{Z}_1$ as output, with $\reg{Y}_1$ being sent to another agent and
    $\reg{Z}_1$ representing a memory register that is retained by the agent
    being described.
    Then $\reg{X}_2$ is received, both $\reg{X}_2$ and the memory register
    $\reg{Z}_1$ are fed into the second channel $\Phi_2$, and so on.}
  \label{fig:strategy}
\end{figure}
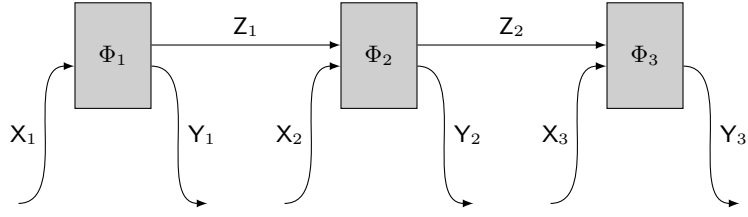

Consider an agent that engages in an interaction involving the exchange of
quantum information with one or more other agents.
Let us suppose, in particular, that the agent being considered first receives
a register $\reg{X}_1$, then sends a register $\reg{Y}_1$, then receives
$\reg{X}_2$, then sends $\reg{Y}_2$, and so on, with its role in the
hypothetical interaction concluding after it receives $\reg{X}_r$ and then
sends $\reg{Y}_r$.
It is to be assumed that the agent may store quantum information between the
rounds of interaction.
Figure~\ref{fig:strategy} depicts the actions of an agent of this sort in the
case $r=3$.
Hereafter we will refer to a network of this form as a \emph{strategy} for the
agent being described.

In the quantum strategies framework, strategies of this sort are represented by
the \emph{Choi representation} of the network.
To be more precise, the network is considered as a single quantum channel
\begin{equation}
  \label{eq:compound-channel}
  \Phi\in\Channel(\X_1\otimes\cdots\otimes\X_r,\Y_1\otimes\cdots\otimes\Y_r),
\end{equation}
with $(\reg{X}_1,\ldots,\reg{X}_r)$ collectively forming the input to this
channel and $(\reg{Y}_1,\ldots,\reg{Y}_r)$ forming the output, and the
Choi representation $J(\Phi)$ is taken as a representation of the strategy.
In general, the Choi representation of a channel taking the form
$\Phi\in\Channel(\X,\Y)$ is given by
\begin{equation}
  J(\Phi) = \sum_{a,b} \Phi\bigl( \ket{a}\bra{b} \bigr) \otimes \ket{a}\bra{b},
\end{equation}
where $a$ and $b$ range over all classical states (or, equivalently, standard
basis elements) of the input space~$\X$, and therefore for $\Phi$ taking the
form \eqref{eq:compound-channel} we have
\begin{equation}
  J(\Phi) \in \Lin(\Y_1\otimes\cdots\otimes\Y_r \otimes 
  \X_1\otimes\cdots\otimes\X_r).
\end{equation}

Not every channel of the form \eqref{eq:compound-channel} can be obtained by
composing channels $\Phi_1,\ldots,\Phi_r$ in the manner just described;
a given channel might not respect the ``time ordering'' in which
each register $\reg{Y}_k$ is produced prior to the registers
$\reg{X}_{k+1},\ldots,\reg{X}_r$ being received.
A necessary and sufficient condition for a channel of the form
\eqref{eq:compound-channel} to decompose into a network of channels
$\Phi_1,\ldots,\Phi_r$ is that its Choi representation is positive
semidefinite,
\begin{equation}
  J(\Phi) \in
  \Pos(\Y_1\otimes\cdots\otimes\Y_r\otimes\X_1\otimes\cdots\otimes\X_r),
\end{equation}
and satisfies a collection of affine linear constraints:
\begin{equation}
  \label{eq:SDP-strategy-constraints}
  \begin{aligned}
    \tr_{\Y_r} \bigl( J(\Phi) \bigr) &= X_{r-1} \otimes \I_{\X_r}\\
    \tr_{\Y_{r-1}} (X_{r-1}) &= X_{r-2} \otimes \I_{\X_{r-1}}\\
    & \;\;\vdots\\
    \tr_{\Y_2} (X_2) &= X_{1} \otimes \I_{\X_2}\\
    \tr_{\Y_1} (X_1) &= \I_{\X_1}
  \end{aligned}
\end{equation}
for $X_1,\ldots,X_{r-1}$ being operators having sizes required by the
equalities.
By the assumption that $J(\Phi)$ is positive semidefinite, the operators
$X_1,\ldots,X_{r-1}$ (if they exist) must be positive semidefinite, and so we
may write
\begin{equation}
  \begin{gathered}
    X_{r-1} \in \Pos(\Y_1\otimes\cdots\otimes\Y_{r-1}
    \otimes\X_1\otimes\cdots\otimes\X_{r-1})\\
    \vdots\\
    X_1\in\Pos(\Y_1\otimes\X_1).
  \end{gathered}
\end{equation}
(The operators $X_{r-1},\ldots,X_1$ happen to be the Choi representations
of the strategies obtained by ``truncating'' the strategy described by the
channels $\Phi_1,\ldots,\Phi_r$, assuming the final memory register is tracing
out in each case.)
It may be noted that, in the case $r=1$, the usual necessary and sufficient
conditions for a map to describe a quantum channel are recovered.

\begin{figure}
  \begin{center} \small
    \begin{tikzpicture}[scale=0.35, 
        turn/.style={draw, minimum height=14mm, minimum width=10mm,
          fill = ChannelColor, text=ChannelTextColor},
        invisibleturn/.style={minimum height=14mm, 
          minimum width=10mm},
        measure/.style={draw, minimum width=7mm, minimum height=7mm,
          fill = ChannelColor},
        >=latex]
      
      \node (V1) at (-8,6) [turn] {$\Phi_1$};
      \node (V2) at (2,6) [turn] {$\Phi_2$};
      \node (V3) at (12,6) [turn] {$\Phi_3$};
      
      \node (M) at (22,0) [measure] {};
      
      \node (P0) at (-13,0) [turn] {$\Psi_1$};
      \node (P1) at (-3,0) [turn] {$\Psi_2$};
      \node (P2) at (7,0) [turn] {$\Psi_3$};
      \node (P3) at (17,0) [turn] {$\Psi_4$};
      
      \node[draw, minimum width=5mm, minimum height=3.5mm, fill=ReadoutColor]
      (readout) at (M) {};
      
      \draw[thick] ($(M)+(0.3,-0.15)$) arc (0:180:3mm);
      \draw[thick] ($(M)+(0.2,0.2)$) -- ($(M)+(0,-0.2)$);
      \draw[fill] ($(M)+(0,-0.2)$) circle (0.5mm);
      
      \draw[->] ([yshift=4mm]V1.east) -- ([yshift=4mm]V2.west)
      node [above, midway] {$\reg{Z}_1$};
      
      \draw[->] ([yshift=4mm]V2.east) -- ([yshift=4mm]V3.west)
      node [above, midway] {$\reg{Z}_2$};
      
      \draw[->] ([yshift=-4mm]V1.east) .. controls +(right:20mm) and 
      +(left:20mm) .. ([yshift=4mm]P1.west) node [right, pos=0.4] 
      {$\reg{Y}_1$};
      
      \draw[->] ([yshift=-4mm]V2.east) .. controls +(right:20mm) and 
      +(left:20mm) .. ([yshift=4mm]P2.west) node [right, pos=0.4] 
      {$\reg{Y}_2$};
      
      \draw[->] ([yshift=-4mm]V3.east) .. controls +(right:20mm) and 
      +(left:20mm) .. ([yshift=4mm]P3.west) node [right, pos=0.4] 
      {$\reg{Y}_3$};
      
      \draw[->] ([yshift=4mm]P0.east) .. controls +(right:20mm) and 
      +(left:20mm) .. ([yshift=-4mm]V1.west) node [left, pos=0.6] 
      {$\reg{X}_1$};
      
      \draw[->] ([yshift=4mm]P1.east) .. controls +(right:20mm) and 
      +(left:20mm) .. ([yshift=-4mm]V2.west) node [left, pos=0.6] 
      {$\reg{X}_2$};
      
      \draw[->] ([yshift=4mm]P2.east) .. controls +(right:20mm) and 
      +(left:20mm) .. ([yshift=-4mm]V3.west) node [left, pos=0.6]
      {$\reg{X}_3$};
      
      \draw[->] (P3.east) -- (M.west);
      
      \draw[->] ([yshift=-4mm]P0.east) -- ([yshift=-4mm]P1.west)
      node [below, midway] {$\reg{W}_1$};
      
      \draw[->] ([yshift=-4mm]P1.east) -- ([yshift=-4mm]P2.west)
      node [below, midway] {$\reg{W}_2$};
      
      \draw[->] ([yshift=-4mm]P2.east) -- ([yshift=-4mm]P3.west)
      node [below, midway] {$\reg{W}_3$};
      
    \end{tikzpicture}
  \end{center}
  \caption{A strategy of the form depicted in Figure~\ref{fig:strategy} may be
    interfaced with the actions of one or more other agents.
    In the interaction pictured, the second agent produces a measurement
    outcome at the conclusion of the interaction.}
  \label{fig:interacting-strategies}
\end{figure}
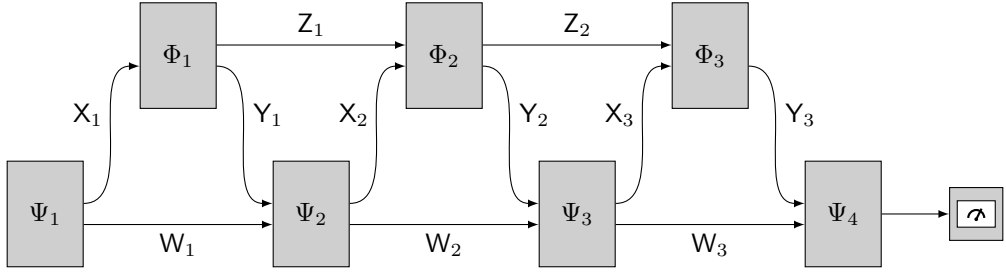

Now suppose that an agent of the form depicted in Figure~\ref{fig:strategy}
interacts with another agent, who performs a measurement at the conclusion of
the interaction, as is suggested by Figure~\ref{fig:interacting-strategies}.
For the sake of clarity, and to connect the framework to the setting of games,
the agent depicted in Figure~\ref{fig:strategy} will be called the
\emph{player} and the new agent will be called the \emph{referee}.
Through the quantum strategies framework, for each possible outcome $a$ that
may be produced by the referee's measurement, one may compute a positive
semidefinite operator
\begin{equation}
  P_a \in \Pos(\Y_1\otimes\cdots\otimes\Y_r \otimes 
  \X_1\otimes\cdots\otimes\X_r)
\end{equation}
with the property that, when the referee and player interact, the probability
for each measurement outcome to appear is given by
\begin{equation}
  \op{Pr}(\text{measurement outcome equals $a$})
  = \bigip{P_a}{Q},
\end{equation}
for $Q = J(\Phi)$ being the representation of the player's strategy.
It is not necessary for the purposes of this paper to explain precisely how
each operator $P_a$ is obtained, except to say that this operator may be
computed efficiently given descriptions of the channels
$\Phi_1,\ldots,\Phi_{r+1}$ and the final measurement.\footnote{%
  The process for obtaining these operators is again based on the Choi
  representation, but one must account for the measurement, the spaces must be
  ordered in a way that matches with the representation $J(\Xi)$, and an
  entry-wise complex conjugation is required to ensure that the expression
  $\bigip{P_a}{Q}$ correctly represents the probability associated with the
  outcome $a$.}

Finally, the quantum strategies framework extends to interactions involving
multiple agents in a fairly straightforward way.
In the context of quantum games, we are interested in interactions in which a
referee, who produces a final measurement outcome at the conclusion of the
interaction, interacts not just with a single player, but with multiple
players.
For example, Figure~\ref{fig:parallel-repetition} depicts the situation in
which a referee, represented by the channels $\Psi_1,\ldots,\Psi_4$ along with
the box suggesting a measurement, interacts with two players, each designated
by a superscript 1 or 2 on their respective channels and the registers they
touch.
In general, such an interaction may involve any number of players~$m$.
Following the standard assumption in non-cooperative game theory, the
$m$ players are assumed to not directly interact with one another:
all interactions are between a player and the referee.
(The referee could choose to pass information from one player to another, but
such an action must be understood as being in accordance with the referee's
specification.)

Also, although Figure~\ref{fig:parallel-repetition} might suggest a symmetry
between the players, this is not required---the registers being exchanged can
have arbitrary size, including the possibility of trivial (dimension 1)
registers that effectively represent the absence of information being sent or
received.
Equivalently, the referee may interleave the messages exchanged with different
players in an arbitrary way, and the number of exchanges may be different with
different players.

In any case of this sort, similar to the single-player case just discussed,
there will always exist an efficiently computable positive semidefinite
operator $P_a$, for each possible outcome of the referee's measurement, for
which the probability associated with that measurement outcome is given by
\begin{equation}
  \label{eq:multiple-players-outcome-probabilities}
  \op{Pr}(\text{measurement outcome equals $a$})
  = \bigip{P_a}{Q_1\otimes\cdots\otimes Q_m},
\end{equation}
assuming that the $m$ players play strategies represented by matrices
$Q_1,\ldots,Q_m$.
Indeed, aside from a permutation of tensor factors, one need not see this
as being an extension of the single-player case at all, for if the $m$ players
do not directly interact, they may be collectively viewed as a single player,
whose representation (again, up to a permutation of tensor factors) is given by
the tensor product $Q_1\otimes\cdots\otimes Q_m$.

\begin{figure}
  \begin{center} \small
    \begin{tikzpicture}[scale=0.35, 
        turn/.style={draw, minimum height=14mm, minimum width=10mm,
          fill = ChannelColor, text=ChannelTextColor},
        invisibleturn/.style={minimum height=14mm, 
          minimum width=10mm},
        measure/.style={draw, minimum width=7mm, minimum height=7mm,
          fill = ChannelColor},
        >=latex]
      
      \node (V1) at (-8,6) [turn] {$\Phi_1^1$};
      \node (V2) at (2,6) [turn] {$\Phi_2^1$};
      \node (V3) at (12,6) [turn] {$\Phi_3^1$};
      
      \node (W1) at (-8,-6) [turn] {$\Phi_1^2$};
      \node (W2) at (2,-6) [turn] {$\Phi_2^2$};
      \node (W3) at (12,-6) [turn] {$\Phi_3^2$};
      
      \node (M) at (22,0) [measure] {};
      
      \node (P0) at (-13,0) [turn] {$\Psi_1$};
      \node (P1) at (-3,0) [turn] {$\Psi_2$};
      \node (P2) at (7,0) [turn] {$\Psi_3$};
      \node (P3) at (17,0) [turn] {$\Psi_4$};
      
      \node[draw, minimum width=5mm, minimum height=3.5mm, fill=ReadoutColor]
      (readout) at (M) {};
      
      \draw[thick] ($(M)+(0.3,-0.15)$) arc (0:180:3mm);
      \draw[thick] ($(M)+(0.2,0.2)$) -- ($(M)+(0,-0.2)$);
      \draw[fill] ($(M)+(0,-0.2)$) circle (0.5mm);
      
      \draw[->] ([yshift=4mm]V1.east) -- ([yshift=4mm]V2.west)
      node [above, midway] {$\reg{Z}_1^1$};
      
      \draw[->] ([yshift=4mm]V2.east) -- ([yshift=4mm]V3.west)
      node [above, midway] {$\reg{Z}_2^1$};
      
      \draw[->] ([yshift=-4mm]V1.east) .. controls +(right:20mm) and 
      +(left:20mm) .. ([yshift=6mm]P1.west) node [right, pos=0.4] 
      {$\reg{Y}_1^1$};
      
      \draw[->] ([yshift=-4mm]V2.east) .. controls +(right:20mm) and 
      +(left:20mm) .. ([yshift=6mm]P2.west) node [right, pos=0.4] 
      {$\reg{Y}_2^1$};
      
      \draw[->] ([yshift=-4mm]V3.east) .. controls +(right:20mm) and 
      +(left:20mm) .. ([yshift=6mm]P3.west) node [right, pos=0.4] 
      {$\reg{Y}_3^1$};
      
      \draw[->] ([yshift=6mm]P0.east) .. controls +(right:20mm) and 
      +(left:20mm) .. ([yshift=-4mm]V1.west) node [left, pos=0.6] 
      {$\reg{X}_1^1$};
      
      \draw[->] ([yshift=6mm]P1.east) .. controls +(right:20mm) and 
      +(left:20mm) .. ([yshift=-4mm]V2.west) node [left, pos=0.6] 
      {$\reg{X}_2^1$};
      
      \draw[->] ([yshift=6mm]P2.east) .. controls +(right:20mm) and 
      +(left:20mm) .. ([yshift=-4mm]V3.west) node [left, pos=0.6]
      {$\reg{X}_3^1$};
      
      \draw[->] ([yshift=-4mm]W1.east) -- ([yshift=-4mm]W2.west)
      node [below, midway] {$\reg{Z}_1^2$};
      
      \draw[->] (P3.east) -- (M.west);
      
      \draw[->] ([yshift=-4mm]W2.east) -- ([yshift=-4mm]W3.west)
      node [below, midway] {$\reg{Z}_2^2$};
      
      \draw[->] ([yshift=4mm]W1.east) .. controls +(right:20mm) and 
      +(left:20mm) .. ([yshift=-6mm]P1.west) node [right,pos=0.4] 
      {$\reg{Y}_1^2$};
      
      \draw[->] ([yshift=4mm]W2.east) .. controls +(right:20mm) and 
      +(left:20mm) .. ([yshift=-6mm]P2.west) node [right,pos=0.4]
      {$\reg{Y}_2^2$};
      
      \draw[->] ([yshift=4mm]W3.east) .. controls +(right:20mm) and 
      +(left:20mm) .. ([yshift=-6mm]P3.west) node [right,pos=0.4]
      {$\reg{Y}_3^2$};
      
      \draw[->] ([yshift=-6mm]P0.east) .. controls +(right:20mm) and 
      +(left:20mm) .. ([yshift=4mm]W1.west) node [left, pos=0.6] 
      {$\reg{X}_1^2$};
      
      \draw[->] ([yshift=-6mm]P1.east) .. controls +(right:20mm) and 
      +(left:20mm) .. ([yshift=4mm]W2.west) node [left, pos=0.6] 
      {$\reg{X}_2^2$};
      
      \draw[->] ([yshift=-6mm]P2.east) .. controls +(right:20mm) and 
      +(left:20mm) .. ([yshift=4mm]W3.west) node [left, pos=0.6]
      {$\reg{X}_3^2$};
      
      \draw[->] ([yshift=0mm]P0.east) -- ([yshift=0mm]P1.west)
      node [below, midway] {$\reg{W}_1$};
      
      \draw[->] ([yshift=0mm]P1.east) -- ([yshift=0mm]P2.west)
      node [below, midway] {$\reg{W}_2$};
      
      \draw[->] ([yshift=0mm]P2.east) -- ([yshift=0mm]P3.west)
      node [below, midway] {$\reg{W}_3$};
      
    \end{tikzpicture}
  \end{center}
  \caption{An interaction between a referee (represented by channels
    $\Psi_1,\ldots,\Psi_4$ along with a measurement) and two players, both
    having a form similar to the strategy pictured in
    Figure~\ref{fig:strategy}.}
  \label{fig:parallel-repetition}
\end{figure}
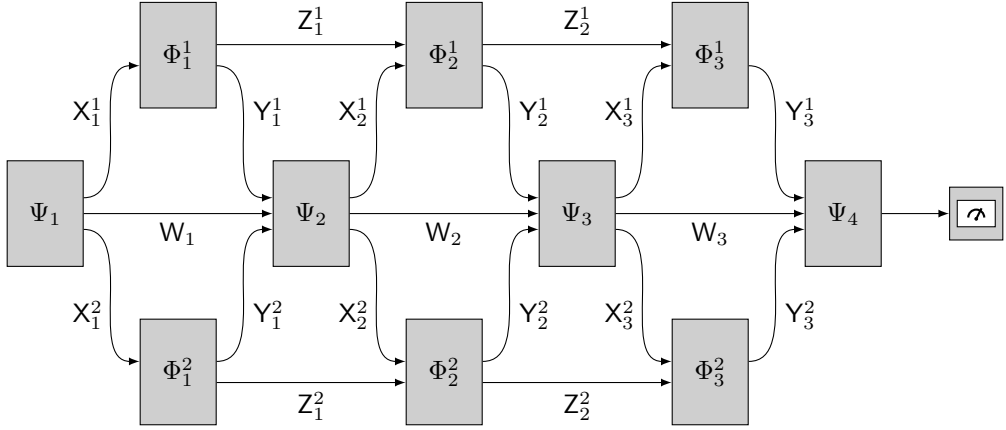

\subsection{PPAD and fixed-point problems}
\label{sec:PPAD}

We now recall the definition of the complexity class $\mathrm{PPAD}$,
which was first defined by Papadimitriou \cite{Papadimitriou1994} to capture
the complexity of certain total functions, including approximate fixed-point
problems when a fixed point is guaranteed to exist.
We also state a result due to Etessami and Yannakakis \cite{EtessamiY2010}
concerning the containment of a specific fixed-point problem in
$\mathrm{PPAD}$ to which we will later reduce the problem of computing
approximate fixed points of functions defined on density operators.

Before proceeding to these definitions, let us remark that all computational
problems in this paper involving real or complex scalars, vectors, matrices,
and so on, are assumed to refer to rational and/or Gaussian rational inputs and
outputs in which the number $a/b$ is encoded as a pair $\langle a,b\rangle$ and
$a/b + i\tsp c/d$ is encoded as a 4-tuples $\langle a,b,c,d\rangle$, for
integers $a$, $b$, $c$, and $d$ represented in signed binary notation.
The length of any such number then refers to the length of the encoding.

\subsubsection*{Total search problems and the complexity class PPAD}

The complexity class $\mathrm{PPAD}$ contains \emph{total search problems}.
In general, a total search problem in the complexity class $\mathrm{TFNP}$
is represented by a collection of sets $\{A_x\,:\,x\in\Sigma^{\ast}\}$,
with $A_x\subseteq\Sigma^{\ast}$ for each $x\in\Sigma^{\ast}$, satisfying these
properties:
\begin{mylist}{8mm}
\item[1.]
  There exists a polynomial $p$ such that $\abs{y}\leq p(\abs{x})$ for every
  $x\in\Sigma^{\ast}$ and $y\in A_x$.
\item[2.]
  There exists a polynomial-time computable predicate $R$ such that
  $R(x,y) = 1$ if and only if $y\in A_x$, for every choice of
  $x,y\in\Sigma^{\ast}$.
\item[3.]
  For every $x\in\Sigma^{\ast}$, the set $A_x$ is non-empty.
\end{mylist}
On a given input string $x\in\Sigma^{\ast}$, the goal of the associated problem
is to find any string $y\in A_x$.
Such search problems are deemed \emph{total} because an acceptable solution is
always guaranteed to exist.

In the context of total search problems in $\mathrm{TFNP}$, it is said that
a problem $\{A_x\,:\,x\in\Sigma^{\ast}\}$ is \emph{polynomial-time reducible}
to another problem $\{B_x\,:\,x\in\Sigma^{\ast}\}$ if there exist
polynomial-time computable functions $f$ and $g$ with the property that for
\begin{equation}
  y\in B_{f(x)} \Rightarrow g(y) \in A_x
\end{equation}
for every string $x\in\Sigma^{\ast}$.
In words, any input to the problem $A = \{A_x\,:\,x\in\Sigma^{\ast}\}$ can be
transformed in polynomial time to an instance $f(x) $ of the problem
$B = \{B_x\,:\,x\in\Sigma^{\ast}\}$ in such a way that any acceptable solution
$y$ to $B$ on input $f(x)$ can be transformed in polynomial time back to an
acceptable solution $g(y)$ to $A$ on input $x$.

Next, to state the definition of the class $\mathrm{PPAD}$, which is contained
in $\mathrm{TFNP}$, we begin with one specific problem in this class called the
\emph{end-of-the-line problem}. 

\pagebreak[3]

\begin{trivlist}
\item \emph{End-of-the-line problem}

\item
  \begin{tabular}{@{}lp{5.2in}@{}}
    Input: & 
    Boolean circuits $P$ and $S$, both having~$n$ input bits and
    $n$ output bits, satisfying $P(0^n) = 0^n \not= S(0^n)$.\\[1mm]
    Output: & 
    Any string $z\in\Sigma^n$ such that $S(P(z)) \not= z \not= 0^n$
    or $P(S(z)) \not= z$.
  \end{tabular}
\end{trivlist}

\noindent
(Formally speaking, if one is given an input string that does not encode
Boolean circuits $P$ and $S$ with the properties indicated, then the
associated set of acceptable solutions is defined as the singleton set
containing the empty string.
Alternatively, one may modify the definition of $\mathrm{TFNP}$ so that
problems may be defined only on a subset of the possible strings.)

The intuition behind this problem is that the circuits $P$ and $S$ allegedly
represent \emph{predecessor} and \emph{successor} functions on the set
$\Sigma^n$.
We envision a graph having vertex set $\Sigma^n$ with a directed edge from $x$
to $y$, for distinct vertices $x$ and $y$, if and only if both $y = S(x)$ and
$x = P(y)$.
The vertex $0^n$ must have in-degree 0 by the assumption $P(0^n) = 0^n$,
meaning that it is a \emph{source}.
The goal is to find either a \emph{sink}, meaning a vertex with out-degree~0,
which must necessarily exist, or a source different from $0^n$.
If $P(S(z))\not=z$, then $z$ is a sink (which could include the possibility
$z = 0^n$ if $0^n$ happens to have out-degree 0), while if
$S(P(z)) \not= z \not= 0^n$ then $z$ is a source different from $0^n$.
It is important that a source different from $0^n$ is an acceptable answer;
the variant of this problem that demands a sink as an output might potentially
be a more difficult computational problem.

Finally, $\mathrm{PPAD}$ is defined as the class of all total search problems
that are polynomial-time reducible to the end-of-the-line problem.

\subsubsection*{Fixed points of barycentric approximations in PPAD}

Suppose that $\{f_x\,:\,x\in\Sigma^{\ast}\}$ is a collection of functions
having the form
\begin{equation}
  f_x : \Delta_n \rightarrow \Delta_n,
\end{equation}
for $n = n(x)$ being polynomially bounded and polynomial-time computable.
We say that $\{f_x\,:\,x\in\Sigma^{\ast}\}$ is a
\emph{polynomial-time computable family} if there exists a polynomial-time
computable function $F$ so that
\begin{equation}
  F(x,\langle v\rangle) = \langle f_x(v) \rangle
\end{equation}
for every rational vector $v\in\Delta_n$, where angled brackets indicate the
encoding of any rational element of $\Delta_n$.

The following theorem follows from a more general result due to Etessami
and Yannakakis \cite{EtessamiY2010}.

\begin{theorem}
  \label{theorem:Etessami-Yannakakis}
  Suppose that $\{f_x\,:\,x\in\Sigma^{\ast}\}$ is a polynomial-time computable
  family of functions having the form
  $f_x : \Delta_n \rightarrow \Delta_n$,
  and for each $x\in\Sigma^{\ast}$ and each positive integer $r$ let
  \begin{equation}
    g_{x,r} : \Delta_n \rightarrow \Delta_n
  \end{equation}
  be the $r$-th order barycentric approximation to $f_x$.
  The problem of computing an exact fixed point of $g_{x,r}$ on the input
  $\langle x, 0^r\rangle$ is contained in the class $\mathrm{PPAD}$.
\end{theorem}


\pagebreak[3]

\section{Definitions of quantum games}
\label{sec:quantum-games}

We will now define a general class of quantum games and state the computational
problem upon which the remainder of the paper focuses.

In the class of games we consider, a referee exchanges quantum registers with
$m$ players over the course of $r$ rounds in a way that generalizes
Figure~\ref{fig:parallel-repetition} (in which $m = 2$ and $r=3$).
The referee's actions are described by channels $\Psi_1,\ldots,\Psi_{r+1}$
along with a final measurement, with these objects taking the following forms.
For $j\in\{2,\ldots,r\}$, the channel $\Psi_j$ takes input registers
\begin{equation}
  \bigl(\reg{W}_{j-1},\reg{Y}_{j-1}^1,\ldots,\reg{Y}_{j-1}^m\bigr)
\end{equation}
and outputs registers
\begin{equation}
  \bigl(\reg{W}_j,\reg{X}_j^1,\ldots,\reg{X}_j^m\bigr).
\end{equation}
The channels $\Psi_1$ and $\Psi_{r+1}$ have a similar form except that
$\Psi_1$ takes no input and $\Psi_{r+1}$ produces a single register
$\reg{W}_{r+1}$ as output.
We note that the registers need not all have the same size, and some may be
trivial, effectively representing the absence of a message transmission.
Finally, the measurement is performed on the register $\reg{W}_{r+1}$ and
has set of outcomes $\Gamma$.
In addition to the referee's actions, as just described, it is to be assumed
that a payoff function $v_k:\Gamma\rightarrow\real$ has been selected for each
player $k\in\{1,\ldots,m\}$.

A referee of this form determines a non-cooperative game, in
which an $m$-tuple of independent strategies for the players is interfaced with
the referee in the most natural way, leading to a distribution over payoffs for
the $m$ players.

Suppose that, by means of the quantum strategies framework, a selection of
the $m$ players' strategies $Q_1,\ldots,Q_m$ has been made, with each
being represented by an operator
\begin{equation}
  Q_k \in \Pos\bigl(
  \Y_1^k\otimes\cdots\otimes\Y_r^k\otimes
  \X_1^k\otimes\cdots\otimes\X_r^k\bigl),
\end{equation}
and suppose moreover that the referee has been represented by a collection
of operators $\{P_a\,:\,a\in\Gamma\}$, as was described in the previous
section.
We then have that each outcome $a\in\Gamma$ is produced by the referee with
probability $\ip{P_a}{Q_1\otimes\cdots\otimes Q_m}$, and the payoffs are then
determined accordingly.
The \emph{expected payoff} for player $k$ is therefore given by
\begin{equation}
  \sum_{a\in\Gamma} v_k(a) \ip{P_a}{Q_1\otimes\cdots\otimes Q_m}
  = \ip{H_k}{Q_1\otimes\cdots\otimes Q_m}
\end{equation}
for
\begin{equation}
  H_k = \sum_{a\in\Gamma} v_k(a) P_a.
\end{equation}
When it is convenient, we will refer to the operators $H_1,\ldots,H_m$ as
\emph{payoff operators}.
We observe that the payoff operators of an interactive quantum game can be
efficiently computed given the description of a referee's actions.

Hereafter let us write
\begin{equation}
  \S_k \subset \Pos\bigl(
  \Y_1^k\otimes\cdots\otimes\Y_r^k\otimes
  \X_1^k\otimes\cdots\otimes\X_r^k\bigl)
\end{equation}
to denote the set of strategy representations available to player $k$, for
each $k\in\{1,\ldots,m\}$.
Each of these sets is bounded and characterized by a finite collection of
affine linear constraints on the positive semidefinite cone acting on the
corresponding spaces.
In particular, the sets $\S_1,\ldots,\S_m$ are convex and compact.

A \emph{Nash equilibrium} of a quantum game of the form being considered is
an $m$-tuple $(Q_1,\ldots,Q_m) \in \S_1\times\cdots\times\S_m$ for which the
equality
\begin{equation}
  \label{eq:Nash-condition}
  \ip{H_k}{Q_1\otimes\cdots\otimes Q_m}
  = \sup_{R\in\S_k}
  \ip{H_k}{Q_1\otimes\cdots\otimes Q_{k-1}\otimes R\otimes Q_{k+1}\otimes
    \cdots \otimes Q_m}
\end{equation}
holds for every $k\in\{1,\ldots,m\}$.
Thus, no player can increase their expected payoff by unilaterally deviating
from a Nash equilibrium $(Q_1,\ldots,Q_m)$.
The existence of a Nash equilibrium in every interactive quantum game follows
from Glicksberg's generalization of Nash's theorem~\cite{Glicksberg1952}.
It is also straightforward to prove the existence of a Nash equilibrium in an
interactive quantum game more directly, through the Kakutani fixed-point
theorem (upon which Glicksberg's generalization is also based), following the
same reasoning as in Nash's proof in \cite{Nash1950} for the existence of an
equilibrium point in classical games.

For any choice of $\varepsilon > 0$, an $m$-tuple
of strategies $(Q_1,\ldots,Q_m) \in \S_1\times\cdots\times\S_m$ is
an \emph{$\varepsilon$-approximate Nash equilibrium} if it is the case that
\begin{equation}
\label{eq:approximate-Nash-condition}
  \ip{H_k}{Q_1\otimes\cdots\otimes Q_m}
  \geq \sup_{R\in\S_k}
  \ip{H_k}{Q_1\otimes\cdots\otimes Q_{k-1}\otimes R\otimes Q_{k+1}\otimes
    \cdots \otimes Q_m} - \varepsilon
\end{equation}
for every $k\in\{1,\ldots,m\}$.
In words, no player can increase their expected payoff by more than
$\varepsilon$ by deviating from an $\varepsilon$-approximate Nash equilibrium
$(Q_1,\ldots,Q_m)$.

For the sake of efficiency, it is prudent at this point to introduce some
additional notation.
For each $k\in\{1,\ldots,m\}$ we define
\begin{equation}
  \label{eq:V_k}
  \V_k = \Y_1^k\otimes\cdots\otimes\Y_r^k\otimes
  \X_1^k\otimes\cdots\otimes\X_r^k,
\end{equation}
so that $\S_k\subset\Pos(\V_k)$, as well as
\begin{equation}
  \V_{-k} =
  \V_1\otimes\cdots\otimes\V_{k-1}\otimes\V_{k+1}\otimes\cdots\otimes\V_m.
\end{equation}
For a given choice of strategies
$(Q_1,\ldots,Q_m)\in\S_1\times\cdots\times\S_m$ we define
\begin{equation}
  Q_{-k} =
  Q_1\otimes\cdots\otimes Q_{k-1} \otimes Q_{k+1}\otimes\cdots\otimes Q_m,
\end{equation}
so that $Q_{-k}\in\Pos(\V_{-k})$.
We stress that this is a tensor product---a similar notation is often used for
Cartesian products.

Observe that, for each $k\in\{1,\ldots,m\}$, there exists a
Hermitian-preserving linear map taking the form
$\Xi_k:\Lin(\V_{-k})\rightarrow\Lin(\V_k)$ and having the property that
\begin{equation}
  \ip{H_k}{Q_1\otimes\cdots\otimes Q_m}
  = \ip{\Xi_k(Q_{-k})}{Q_k}
\end{equation}
for every choice of $(Q_1,\ldots,Q_m) \in \S_1\times\cdots\times\S_m$
(or indeed for any choice of Hermitian operators $Q_1,\ldots,Q_m$, not just
strategies).
Explicitly,
\begin{equation}
  \Xi_k(Q_{-k})
  = \tr_{\V_{-k}}\bigl(
  (Q_1\otimes\cdots\otimes Q_{k-1} \otimes\I_{\V_k}
  \otimes Q_{k+1} \otimes \cdots\otimes Q_m)H_k\bigr).
\end{equation}
An equivalent condition to \eqref{eq:Nash-condition} is then that
\begin{equation}
  \ip{\Xi_k(Q_{-k})}{Q_k} = \sup_{R\in\S_k} \ip{\Xi_k(Q_{-k})}{R},
\end{equation}
while \eqref{eq:approximate-Nash-condition} is equivalent to
\begin{equation}
  \ip{\Xi_k(Q_{-k})}{Q_k} \geq \sup_{R\in\S_k} \ip{\Xi_k(Q_{-k})}{R} -
  \varepsilon.
\end{equation}

We may now define the computational problem of approximating a Nash equilibrium
of a quantum game.
We assume that the input to the problem consists of the payoff operators of a
given game, along with positive real number $\varepsilon$, although as noted
above one could alternatively describe a quantum game in terms of the referee's
actions, from which the payoff operators may be computed.

\begin{trivlist}
\item \emph{Approximate quantum Nash equilibrium}

\item
  \begin{tabular}{@{}lp{5.2in}@{}}
    Input: &
    Hermitian operators
    $H_1,\ldots,H_m \in \Herm(\V_1\otimes\cdots\otimes \V_m)$, for each $\V_k$
    taking the form \eqref{eq:V_k}, along with a positive real number
    $\varepsilon$.\\
    Output: & 
    An $\varepsilon$-approximate Nash equilibrium $(Q_1,\ldots,Q_m)$ of the
    interactive quantum game described by $H_1,\ldots,H_m$.
  \end{tabular}
\end{trivlist}

The following theorem, which is proved in the next section, represents the main
result of this paper.

\begin{theorem}
  \label{theorem:approximate-quantum-Nash-in-PPAD}
  The problem of computing an approximate quantum Nash equilibrium is
  contained in the complexity class PPAD.
\end{theorem}


\section{Approximate quantum Nash equilibria in PPAD}

The purpose of this section is to prove
Theorem~\ref{theorem:approximate-quantum-Nash-in-PPAD}.
We shall begin with an overview of the proof, followed by three subsections
that address specific aspects of it.

The proofs of the existence of Nash equilibria in interactive quantum games
suggested above are both based on the Kakutani fixed-point theorem.
Toward the goal of establishing that the problem of approximating Nash
equilibria in quantum games is in the complexity class $\mathrm{PPAD}$,
however, it is instructive to consider a different path, based on an extension
of Nash's 1951 proof \cite{Nash1951} of the existence of equilibria in
classical games, which makes use of the Brouwer fixed-point theorem together
with the notion of a gain function.
This is a familiar path to analogous results in the classical setting
\cite{Papadimitriou1994,DaskalakisGP2009,ChenDT2009,EtessamiY2010}.

Our first step is to prove that the problem of approximating fixed points of
a certain class of continuous functions defined on \emph{density operators} is
contained in $\mathrm{PPAD}$.
This is done by means of a reduction, based on the discrete Wigner
representation defined in Section~\ref{sec:discrete-Wigner}, to the fixed-point
problem on the simplex established to be in $\mathrm{PPAD}$ by
Theorem~\ref{theorem:Etessami-Yannakakis}.

The second step is to consider an interactive quantum generalization of Nash's
gain function.
Intuitively speaking, this is a function defined on $m$-tuples of strategies
that improves each player's strategy, relative to the other players' strategies
being considered, so that the fixed points of this function are equilibrium
points.
This allows for the reduction of the problem of finding an approximate Nash
equilibrium in a quantum game to finding an approximate fixed point of this
gain function, which may be expressed as a function on density operators.

The computations required by both of the steps just described cannot be
performed exactly using rational number computations.
To control the precision required by rational number approximations to these
computations, we must bound the \emph{Lipschitz moduli} of various functions
that are composed to obtain the reduction.
This includes functions expressible as semidefinite programs but not known to
have closed form expressions.

The subsections that follow address these aspects of the proof.
The first subsection is concerned entirely with the Lipschitz moduli of various
function that will be needed in the remaining subsections, establishing bounds
that allow the proof to go through.
The second subsection establishes that the problem of computing approximate
fixed points of continuous functions defined on density operators (or Cartesian
products of density operators) is contained in PPAD.
And finally, the third subsection reduces the problem of computing approximate
Nash equilibria of interactive quantum games to the problem of computing
fixed points of continuous functions on density operators.


\subsection{Some useful functions and bounds on their Lipschitz moduli}
\label{sec:Lipschitz-functions}

This subsection simply lists several functions relevant to the proof together
with bounds on their Lipschitz moduli.

Whenever we refer to the Lipschitz condition for any function, defined for
vectors or operators, we will always use the 2-norm, meaning the standard
Euclidean norm for $\real^n$ and the Frobenius norm for the $n\times n$ complex
Hermitian operators $\Herm(\complex^n)$.
That is, a function $f$ is $K$-Lipschitz if
\begin{equation}
  \norm{f(u) - f(v)}_2 \leq K \tsp \norm{u - v}_2
\end{equation}
for all vectors $u$ and $v$ on which it is defined, and likewise for functions
defined on operators rather than vectors.
We refer to $K$ as the \emph{Lipschitz modulus} of $f$, as opposed to the more
standard \emph{Lipschitz constant}, as $K$ will generally not be constant
(as a function of the input length) for the functions we will encounter.

\subsubsection*{The discrete Wigner representation.}

The function $\psi:\Herm(\complex^n)\rightarrow\real^{n^2}$ associated
with the discrete Wigner representation we have defined
is $(1/K)$-Lipschitz, while $\psi^{-1}$ is $K$-Lipschitz,
for $K = \sqrt{n}\tsp(n+1)$.
More precisely, by the orthogonality of the operators
$\{V_1,\ldots,V_{n^2}\}$, we have the equality conditions
\begin{equation}
  \norm{\psi(H) - \psi(K)}_2 = \frac{1}{\sqrt{n}(n+1)}\norm{H - K}_2
\end{equation}
and
\begin{equation}
  \bignorm{\psi^{-1}(u) - \psi^{-1}(v)}_2 = \sqrt{n}(n+1)\norm{u-v}_2.
\end{equation}
These two moduli will cancel one another in the analysis to follow in the
next subsection.

\subsubsection*{Tensor products of density operators.}

The tensor product mapping
\begin{equation}
  (\rho_1,\ldots,\rho_m) \mapsto \rho_1\otimes\cdots\otimes\rho_m,
\end{equation}
from $\Density(\complex^{n_1})\times\cdots\times\Density(\complex^{n_m})$
to $\Density(\complex^{n_1}\otimes\cdots\otimes\complex^{n_m})$, is
$\sqrt{m}$-Lipschitz:
\begin{equation}
  \begin{aligned}
    \norm{\rho_1\otimes\cdots\otimes\rho_m
      - \sigma_1\otimes\cdots\otimes\sigma_m}_2
    \hspace{-4cm} \\[1mm]
    & \leq \norm{\rho_1\otimes\rho_2\otimes\cdots\otimes\rho_m-
      \sigma_1\otimes\rho_2\otimes\cdots\otimes\rho_m}_2\\
    & \quad + \norm{\sigma_1\otimes\rho_2\otimes\cdots\otimes\rho_m
      - \sigma_1\otimes\sigma_2\otimes\cdots\otimes\sigma_m}_2\\[1mm]
    & \leq \norm{\rho_1 - \sigma_1}_2
    \norm{\rho_2\otimes\cdots\otimes\rho_m}_2 \\
    & \quad + \norm{\sigma_1}_2 \norm{\rho_2\otimes\cdots\otimes\rho_m - 
      \sigma_2\otimes\cdots\otimes\sigma_m}_2 \\[1mm]
    & \leq \norm{\rho_1 - \sigma_1}_2 +
    \norm{\rho_2\otimes\cdots\otimes\rho_m
      - \sigma_2\otimes\cdots\otimes\sigma_m}_2,
  \end{aligned}
\end{equation}
and by iterating,
\begin{equation}
  \begin{aligned}
    \norm{\rho_1\otimes\cdots\otimes\rho_m
      - \sigma_1\otimes\cdots\otimes\sigma_m}_2 \hspace{-4cm}\\
    & \leq \norm{\rho_1 - \sigma_1}_2 + \cdots +
    \norm{\rho_m - \sigma_m}_2\\
    & \leq \sqrt{m}\,\bignorm{(\rho_1,\ldots,\rho_m) -
      (\sigma_1,\ldots,\sigma_m)}_2.
  \end{aligned}
\end{equation}

\subsubsection*{The maps $\Xi_k$.}

Recall the maps $\Xi_k:\Lin(\V_{-k})\rightarrow\Lin(\V_k)$ defined in the
previous section, which satisfy
\begin{equation}
  \ip{H_k}{Q_1\otimes\cdots\otimes Q_m} = \ip{\Xi_k(Q_{-k})}{Q_k}
\end{equation}
for all Hermitian operators $Q_1,\ldots,Q_m$, where $H_k$ is considered to
be fixed.
Explicitly, 
\begin{equation}
  \Xi_k(Q_{-k})
  = \tr_{\V_{-k}}\bigl(
  (Q_1\otimes\cdots\otimes Q_{k-1} \otimes\I_{\V_k}
  \otimes Q_{k+1} \otimes \cdots\otimes Q_m)H_k\bigr).
\end{equation}
Let us write $n_k = \dim(\V_k)$ and $n = n_1 \cdots n_m$.

First, the mapping
\begin{equation}
  Q_{-k} \mapsto
  Q_1\otimes\cdots\otimes Q_{k-1}\otimes\I_{\V_k}\otimes Q_{k+1}
  \otimes\cdots\otimes Q_m
\end{equation}
is $\sqrt{n_k}$-Lipschitz, while right-multiplication by $H_k$ is
$\norm{H_k}$-Lipschitz, $\norm{H_k}$ denoting the spectral norm of $H_k$.

Next, every quantum channel $\Psi:\Lin(\X) \rightarrow \Lin(\Y)$, including
the trace, is 1-Lipschitz with respect to the trace norm, and therefore
\begin{equation}
  \norm{\Psi(X) - \Psi(Y)}_2
  \leq \norm{\Psi(X) - \Psi(Y)}_1
  \leq \norm{X - Y}_1
  \leq \sqrt{\dim(\X)} \, \norm{X - Y}_2.
\end{equation}
That is, every channel is $\sqrt{\dim(\X)}$-Lipschitz with respect to
the Frobenius norm, for~$\X$ being the input space of the channel.
We may also note that tensoring any linear map (whether a channel or not)
with the identity channel does not change its Lipschitz modulus.
It follows that the partial trace over the space $\V_{-k}$ has Lipschitz
modulus $\sqrt{n/n_k}$.

Composing these functions, we find that the mapping $\Xi_k$ is
$(\norm{H_k}\sqrt{n})$-Lipschitz.

\subsubsection*{Projections onto closed and convex sets.}

For any closed and convex set $\C$, we define
$\op{proj}(X\,|\,\C)$ to be the projection of $X$ onto the set $\C$, 
meaning the unique point contained in $\C$ that is closest to $X$ with
respect to the 2-norm (or Frobenius norm).
The function $X\mapsto\op{proj}(X\,|\,\C)$ is, as is well known, always
1-Lipschitz.

\subsubsection*{Normalizing positive semidefinite operators.}

Next, define a function that \emph{normalizes} any positive semidefinite
operator $P\in\Pos(\complex^n)$ in the following way:
\begin{equation}
  \op{normalize}(P) = 
  \begin{cases}
    \frac{P}{\tr(P)} & \tr(P) \geq 1\\[1mm]
    P + (1 - \tr(P))\frac{\I_n}{n} & \tr(P) < 1.
  \end{cases}
\end{equation}
Strictly speaking this may not really be a normalization in the case that
$\tr(P) < 1$, but this function serves our purposes nevertheless.

The function $\op{normalize}:\Pos(\complex^n)\rightarrow\Density(\complex^n)$
is $(4n)$-Lipschitz.
This is perhaps easiest to prove by expressing the function as
$\mathrm{normalize} = g\circ f$ where $f$ and $g$ are defined as follows:
\begin{equation}
  \begin{aligned}
    f(P) & =
    \begin{cases}
      P & \tr(P) \geq 1\\[1mm]
      P + (1 - \tr(P))\frac{\I_n}{n} & \tr(P) < 1,
    \end{cases}\\[2mm]
    g(P) & = 
    \begin{cases}
      \frac{P}{\tr(P)} & \tr(P) \geq 1\\[1mm]
      P & \tr(P) < 1.
    \end{cases}
  \end{aligned}
\end{equation}

The function $f$ is $(2\sqrt{n})$-Lipschitz, which may be established by
considering three cases.
If $\tr(P)\geq 1$ and $\tr(Q)\geq 1$, then 
$\norm{f(P) - f(Q)}_2 = \norm{P - Q}_2$, trivially.
If $\tr(P)\geq 1$ and $\tr(Q)< 1$, then
\begin{equation}
  \begin{multlined}
    \norm{f(P) - f(Q)}_1
    = \Bignorm{P - Q - (1 - \tr(Q))\frac{\I_n}{n}}_1
    \leq \norm{P - Q}_1 + (1 - \tr(Q))\\[1mm]
    \leq \norm{P - Q}_1 + (\tr(P) - \tr(Q))
    \leq 2 \norm{P - Q}_1
  \end{multlined}
\end{equation}
and therefore
\begin{equation}
  \norm{f(P) - f(Q)}_2
  \leq 2 \norm{P - Q}_1 \leq 2\sqrt{n}\norm{P - Q}_2.
\end{equation}
If $\tr(P)<1$ and $\tr(Q)< 1$, then
\begin{equation}
  \begin{multlined}
    \norm{f(P) - f(Q)}_1
    = \Bignorm{P - Q - (\tr(P) - \tr(Q))\frac{\I_n}{n}}_1\\[1mm]
    \leq \norm{P - Q}_1 + \abs{\tr(P) - \tr(Q)}
    \leq 2 \norm{P - Q}_1,
  \end{multlined}
\end{equation}
and so again
\begin{equation}
  \norm{f(P) - f(Q)}_2
  \leq 2 \norm{P - Q}_1 \leq 2\sqrt{n}\norm{P - Q}_2.
\end{equation}

On the set of positive semidefinite operators having trace at least one,
the function $g$ is $(1 + \sqrt{n})$-Lipschitz; supposing that
$P,Q\in\Pos(\complex^n)$ satisfy $\tr(Q)\geq 1$ and $\tr(P)\geq 1$, we
find that
\begin{equation}
  \begin{aligned}
    \biggnorm{\frac{P}{\tr(P)} - \frac{Q}{\tr(Q)}}_2
    & \leq 
    \biggnorm{\frac{P}{\tr(P)} - \frac{Q}{\tr(P)}}_2 + 
    \biggnorm{\frac{Q}{\tr(P)} - \frac{Q}{\tr(Q)}}_2\\
    & =
    \frac{\norm{P - Q}_2}{\tr(P)}
    + \biggl( \frac{\tr(Q) - \tr(P)}{\tr(P)\tr(Q)}\biggr) \norm{Q}_2\\
    & \leq (1 + \sqrt{n})\norm{P - Q}_2.
  \end{aligned}
\end{equation}

Using $2\sqrt{n}(\sqrt{n}+1) \leq 4 n$ we obtain that $\mathrm{normalize}$ is
$(4n)$-Lipschitz.

\subsection{Fixed points of functions on density operators}
\label{sec:Density-matrix-fixed_point}

We now consider the computational problem of approximating fixed points of
continuous functions defined on density operators, proving that this problem is
in $\mathrm{PPAD}$ for functions having exponentially bounded Lipschitz moduli.

To state this fact more precisely, we require a few definitions.
First, a density operator $\rho\in\Density(\complex^n)$ is an
\emph{$\varepsilon$-approximate fixed point} of a function
$f:\Density(\complex^n) \rightarrow \Density(\complex^n)$ provided that
\begin{equation}
  \norm{f(\rho) - \rho}_2 \leq \varepsilon.
\end{equation}

Next, suppose that $\{f_x\,:\,x\in\Sigma^{\ast}\}$ is a collection of functions
having the form
\begin{equation}
  f_x : \Density(\complex^n) \rightarrow \Density(\complex^n)
\end{equation}
for $n = n(x)$ being polynomially bounded.
Mirroring a definition from Section~\ref{sec:PPAD} for functions defined on the
unit simplex, we shall say that $\{f_x\,:\,x\in\Sigma^{\ast}\}$ is a
\emph{polynomial-time computable family} if there exists a polynomial-time
computable function $F$ so that
\begin{equation}
  F(x,\langle \rho\rangle) = \langle f_x(\rho) \rangle
\end{equation}
for every rational density operator $\rho\in\Density(\complex^n)$, with
angled brackets representing encodings of rational density operators.
We must also define an approximate variant of this notion:
$\{f_x\,:\,x\in\Sigma^{\ast}\}$ is a
\emph{polynomial-time approximable family} if there exists a polynomial-time
computable family $\{g_{x,\varepsilon}\}$ satisfying
\begin{equation}
  \bignorm{
    f_x - g_{x,\varepsilon}
  }_2 \leq \varepsilon
\end{equation}
for every $x\in\Sigma^{\ast}$ and every positive rational number
$\varepsilon$.

Finally, the problem of computing $\varepsilon$-approximate fixed points of the
family $\{f_x\}$ is to output the encoding of any $\varepsilon$-approximate
fixed point of the function $f_x$ on input~$(x,\varepsilon)$.

\begin{theorem}
  \label{theorem:density-operator-fixed-points}
  Let $\{f_x\}$ be a polynomial-time approximable family of functions
  on density operators, let $p$ be a polynomial, and assume that each
  function $f_x$ is $K_x$-Lipschitz, for $K_x = 2^{p(\abs{x})}$.
  The problem of computing $\varepsilon$-approximate fixed points of the
  family $\{f_x\}$ is in $\mathrm{PPAD}$.
\end{theorem}

\begin{proof}
  Let us first observe that there is no loss of generality in assuming
  that, for every input $x$, the dimension $n$ is odd and at least~3.
  The case $n=1$ is trivial, and if $n$ is even, one may substitute the
  function $f_x$ by 
  $h_x : \Density(\complex^{n+1}) \rightarrow \Density(\complex^{n+1})$
  defined as
  \begin{equation}
    h_x
    \begin{pmatrix}
      P & u\\
      u^{\ast} & \lambda
    \end{pmatrix}
    = 
    \begin{pmatrix}
      f_x\bigl(P + \lambda\frac{\I}{n}\bigr) & 0\\
      0 & 0
    \end{pmatrix}.
  \end{equation}
  The Lipschitz modulus of $h_x$ is at most $\sqrt{2}$ times that of $f_x$,
  and every fixed point of $h_x$ takes the form
  \begin{equation}
    \sigma = 
    \begin{pmatrix}
      \rho & 0\\
      0 & 0
    \end{pmatrix}
  \end{equation}
  for $\rho$ a fixed point of $f_x$.
  If
  \begin{equation}
    \begin{pmatrix}
      P & u\\
      u^{\ast} & \lambda
    \end{pmatrix}
  \end{equation}
  is an $\varepsilon$-approximate fixed point of $h_x$, then it follows that
  \begin{equation}
    2 \norm{u}^2 + \lambda^2
    =
    \biggnorm{
      \begin{pmatrix}
        0 & u\\
        u^{\ast} & \lambda
      \end{pmatrix}
    }_2^2
    \leq
    \biggnorm{
      h_x
      \begin{pmatrix}
        P & u\\
        u^{\ast} & \lambda
      \end{pmatrix}
      -
      \begin{pmatrix}
        P & u\\
        u^{\ast} & \lambda
      \end{pmatrix}
    }_2^2
    \leq \varepsilon^2,
  \end{equation}
  from which it follows that
  \begin{equation}
    \begin{multlined}
      \biggnorm{f_x\biggl(P + \frac{\lambda \I}{n}\biggr)
        - \Bigl(P + \frac{\lambda \I}{n}\Bigr)}_2\\[2mm]
      \leq
      \biggnorm{
        h_x
        \begin{pmatrix}
          P & u\\
          u^{\ast} & \lambda
        \end{pmatrix}
        -
        \begin{pmatrix}
          P & u\\
          u^{\ast} & \lambda
        \end{pmatrix}
      }_2
      + 
      \biggnorm{
        \begin{pmatrix}
          P & u\\
          u^{\ast} & \lambda
        \end{pmatrix}
        -
        \begin{pmatrix}
          P+\frac{\lambda\I}{n} & 0\\
          0 & 0
        \end{pmatrix}
      }_2
      \leq \Bigl(1 + \sqrt{2}\Bigr)\varepsilon.
    \end{multlined}
  \end{equation}
  Thus, an $\varepsilon$-approximate fixed point of $f_x$ is easily obtained
  from an $(\varepsilon/3)$-approximate fixed point of $h_x$.
  
  Assuming now that $n$ is odd for each $x$, we define a function
  $g_x : \Delta_{n^2} \rightarrow \Delta_{n^2}$ as
  \begin{equation}
    g_x(v) = \psi\bigl( f_x \bigl(
    \op{proj}\bigl(\psi^{-1}(v)\,\big|\,\Density(\complex^n)\bigr)\bigr)\bigr)
  \end{equation}
  Here, the projection function is as defined in the previous subsection and
  $\psi$ is the mapping associated with the discrete Wigner representation
  defined in Section~\ref{sec:discrete-Wigner}.
  Given that $f_x$ is $K_x$-Lipschitz, it follows that $g_x$ is
  $K_x$-Lipschitz as well, as the projection is 1-Lipschitz and the Lipschitz
  moduli of the discrete Wigner mappings cancel.  

  Given that $f_x$ is polynomial-time approximable, it is possible to compute,
  in polynomial time, an approximation $\widetilde{g_x}$ to $g_x$ satisfying
  \begin{equation}
    \norm{\widetilde{g_x}(u) - g_x(u)}_2 \leq \frac{\varepsilon}{16 n^2 K_x}
  \end{equation}
  for every rational vector $u\in\Delta_{n^2}$.
  We note, in particular, that the projection onto $\Density(\complex^n)$
  may be approximated by first approximating a spectral decomposition of the
  operator $\psi^{-1}(v)$ and then projecting its eigenvalues onto the unit
  simplex $\Delta_n$.
  Alternatively, this projection arises as a special case of one
  discussed in the next subsection, where the ellipsoid method provides a
  polynomial-time algorithm to approximate the projection.

  Next, set 
  \begin{equation}
    r = (n^2+1)
    \bigceil{\log(1/\varepsilon) + 2 p(\abs{x}) + 2\log(n) + 4}.
  \end{equation}
  This number is polynomial in $\abs{x}$ and $\log(1/\varepsilon)$, and
  has been selected so that
  \begin{equation}
    \Bigl(1 - \frac{1}{n^2 + 1}\Bigr)^r
    < \exp\bigl(-\log(1/\varepsilon) - 2 p(\abs{x}) - 2 \log(n) - 4)
    < \frac{\varepsilon}{16 n^2 K_x^2}.
  \end{equation}
  By Theorem~\ref{theorem:Etessami-Yannakakis}, one may therefore compute an
  exact fixed point $v\in\Delta_{n^2}$ of the $r$-th level barycentric
  approximation to $\widetilde{g_x}$ in $\mathrm{PPAD}$.

  Supposing that such a fixed point $v$ is expressed as a convex combination
  \begin{equation}
    v = q_1 v_1 + \cdots + q_{n^2} v_{n^2}
  \end{equation}
  for $v_1,\ldots,v_{n^2} \in \B_{n^2}^r$ denoting vertices in any one of the
  simplices constructed at the \mbox{$r$-level} of the barycentric subdivision,
  we find that
  \begin{equation}
    \begin{aligned}
      \norm{g_x(v) - v}_2
      & = \bignorm{{g_x}(v) -
        \bigl(q_1 \widetilde{g_x}(v_1) + \cdots
        q_{n^2} \widetilde{g_x}(v_{n^2})\bigr)}_2\\
      & \leq
      \sum_{j = 1}^{n^2}
      q_j \bignorm{g_x(v) - \widetilde{g_x}(v_j)}_2\\
      & \leq
      \sum_{j = 1}^{n^2}
      q_j
      \Bigl(
      \bignorm{g_x(v) - g_x(v_j)}_2 + \bignorm{
        g_x(v_j) - \widetilde{g_x}(v_j)}_2 \Bigr)\\
      & \leq K_x \Bigl(1 - \frac{1}{n^2 + 1}\Bigr)^r
      + \frac{\varepsilon}{16 n^2 K_x}\\
      & \leq \frac{\varepsilon}{8 n^2 K_x}.
    \end{aligned}
  \end{equation}
  Thus, $v$ is an $(\varepsilon/(8 n^2 K_x))$-approximate fixed point of $g_x$.

  Now consider the Hermitian operator $\psi^{-1}(v)$.
  We have
  \begin{equation}
    \bignorm{ \psi^{-1}(v) - \psi^{-1}(g(v))}_2
    = \sqrt{n}(n+1) \bignorm{v - g(v)}_2 \leq \frac{\varepsilon}{4 K_x},
  \end{equation}
  and given that $\psi^{-1}(g(v))$ is necessarily a density operator, the
  operator $\psi^{-1}(v)$ therefore has distance at most $\varepsilon/(4 K_x)$
  from the set of density operators.
  By computing a density operator $\rho$ satisfying
  \begin{equation}
    \bignorm{
      \rho - \op{proj}\bigl(\psi^{-1}(v)\, \big| \,
      \Density(\complex^n)\bigr)}_2
    \leq \frac{\varepsilon}{4 K_x}
  \end{equation}
  as suggested above, we therefore have
  \begin{equation}
    \bignorm{
      \rho - \psi^{-1}(v)}_2 \leq \frac{\varepsilon}{2 K_x}.
  \end{equation}
  Consequently, noting that
  \begin{equation}
    f\bigl( \op{proj} \bigl(\psi^{-1}(v) \,\big|\, \Density(\complex^n)
    \bigr) = \psi^{-1}(g(v)),
  \end{equation}
  we find, by the triangle inequality, that
  \begin{equation}
    \begin{multlined}
      \bignorm{f_x(\rho) - \rho}_2
      \leq
      \bignorm{
        f(\rho) - f \bigl(
        \op{proj}\bigl(\psi^{-1}(v)\,\big|
        \,\Density(\complex^n)\bigr)\bigr)}_2\\[2mm]     
      + \bignorm{\psi^{-1}(g(v)) - \psi^{-1}(v)}_2
      + \bignorm{\psi^{-1}(v) - \rho}_2
      \leq \frac{\varepsilon}{4} + \frac{\varepsilon}{4 K_x} +
      \frac{\varepsilon}{2 K_x} \leq \varepsilon.
    \end{multlined}
  \end{equation}
  Thus, $\rho$ is an $\varepsilon$-approximate fixed point of $f_x$.
  As $\rho$ has been computed in polynomial time from $v$, the theorem is
  proved.
\end{proof}

\begin{cor}
  \label{cor:fixed-point}
  Let $\{f_x\}$ be a polynomial-time approximable family of functions having
  the form
  \begin{equation}
    f_x :
    \Density\bigl(\complex^{n_1}\bigr) \times \cdots \times
    \Density\bigl(\complex^{n_m}\bigr)
    \rightarrow
    \Density\bigl(\complex^{n_1}\bigr) \times \cdots \times
    \Density\bigl(\complex^{n_m}\bigr),
  \end{equation}
  for positive integers $n_1,\ldots,n_m$,
  let $p$ be a polynomial,
  and assume that each function $f_x$ is $K_x$-Lipschitz,
  for $K_x = 2^{p(\abs{x})}$.
  The problem of computing $\varepsilon$-approximate fixed points of the
  family $\{f_x\}$ is in $\mathrm{PPAD}$.
\end{cor}

\begin{proof}
  Let $n = n_1 + \cdots + n_m$ and define a mapping
  $h:\Density(\complex^n)\rightarrow\Density(\complex^n)$
  as follows:
  \begin{equation}
    h
    \begin{pmatrix}
      X_{1,1} & \cdots & X_{1,m}\\
      \vdots & \ddots & \vdots\\
      X_{m,1} & \cdots & X_{m,m}
    \end{pmatrix}
    =
    \frac{1}{m}
    \begin{pmatrix}
      \op{normalize}(m X_{1,1}) & & 0 \\
       & \ddots & \\
       0 & & \op{normalize}(m X_{m,m})
    \end{pmatrix},
  \end{equation}
  where it is to be understood that each $X_{i,j}$ has $n_i$ rows and $n_j$
  columns.
  The mapping $h$ is $(4n)$-Lipschitz and projects onto
  operators having the form
  \begin{equation}
    \label{eq:block-state}
    \frac{1}{m}
    \begin{pmatrix}
      \rho_1 & & 0\\
      & \ddots & \\
      0 & & \rho_m
    \end{pmatrix}.
  \end{equation}
  By composing $f_x$ with $h$ in the natural way, one obtains a function
  $g_x:\Density(\complex^n)\rightarrow\Density(\complex^n)$ such that
  \begin{equation}
    g_x
    \begin{pmatrix}
      X_{1,1} & \cdots & X_{1,m}\\
      \vdots & \ddots & \vdots\\
      X_{m,1} & \cdots & X_{m,m}
    \end{pmatrix}
    =
    \frac{1}{m}
    \begin{pmatrix}
      \sigma_1 & & 0 \\
       & \ddots & \\
       0 & & \sigma_m
    \end{pmatrix}
  \end{equation}
  for
  \begin{equation}
    (\sigma_1,\ldots,\sigma_m)
    =
    f_x(\op{normalize}(m X_{1,1}),\ldots,\op{normalize}(m X_{m,m})).
  \end{equation}
  
  Finally, from any approximate fixed point of the family $\{g_x\}$,
  an $\varepsilon$-approximate fixed point for $\{f_x\}$ is obtained
  by applying to it the function $h$ and reading off the diagonal operators.
  The problem of approximating fixed points of $\{f_x\}$ therefore reduces in
  polynomial time to that of $\{g_x\}$, which is in the class $\mathrm{PPAD}$.
\end{proof}


\subsection{Nash equilibria as fixed points of functions}

The final step of the proof of
Theorem~\ref{theorem:approximate-quantum-Nash-in-PPAD} is to reduce the problem
of computing approximate Nash equilibria of interactive quantum games to the
approximate fixed-point problem on Cartesian products of density operators
established by Corollary~\ref{cor:fixed-point} to be in PPAD.
To do this, we will consider an extension of Nash's gain function to quantum
strategies, as they are represented within the quantum strategies framework. 

For a quantum game of the general form described in
Section~\ref{sec:quantum-games}, the set of strategies available each player
$k\in\{1,\ldots,m\}$ is represented by the set
\begin{equation}
  \S_k \subset \Pos\bigl(\Y_1^k\otimes\cdots\otimes\Y_r^k\otimes
  \X_1^k\otimes\cdots\otimes\X_r^k\bigr),
\end{equation}
and we observe that for every choice of $Q_k\in\S_k$ we have
\begin{equation}
  \tr(Q_k) = d_k \stackrel{\text{\tiny def}}{=}
  \dim\bigl(\X_1^k\otimes\cdots\otimes\X_r^k\bigr).
\end{equation}
Define the set
\begin{equation}
  \C_k = \frac{1}{d_k}\S_k \subseteq \Density(\V_k),
\end{equation}
as well as the cone
\begin{equation}
  \K_k = \op{cone}(\C_k) = \bigl\{
  \lambda \rho\,:\, \lambda \geq 0,\;
  \rho\in\C_k\bigr\}.
\end{equation}

Now, for a given $m$-tuple $(\rho_1,\ldots,\rho_m)$ of density operators, we
define $G(\rho_1,\ldots,\rho_m)$ in the following way.
First, for each $k\in\{1,\ldots,m\}$, define
\begin{equation}
  \begin{aligned}
    \sigma_k & = \op{proj}(\rho_k\,|\,\C_k),\\
    \alpha_k & = \bigip{\Xi_k(\sigma_{-k})}{\sigma_k},\\
    P_k & = \op{proj}\bigl(
    \Xi_k(\sigma_{-k}) - \alpha_k \I_{\V_k}\,|\,\K_k),
  \end{aligned}
\end{equation}
and
\begin{equation}
  G_k(\rho_1,\ldots,\rho_m) = \op{normalize}\bigl(\sigma_k + P_k)
  = \frac{\sigma_k + P_k}{1 + \tr(P_k)}.
\end{equation}
Then define
\begin{equation}
  G(\rho_1,\ldots,\rho_m)
  = \bigl(G_1(\rho_1,\ldots,\rho_m),\ldots,G_m(\rho_1,\ldots,\rho_m)\bigr).
\end{equation}

By combining the Lipschitz moduli for the functions from which $G$ is formed,
over-estimating for the sake of a simple expression, we have that $G$ is
$K$-Lipschitz for
\begin{equation}
  K = 4n^2m M,\quad
  M = \max\{\norm{H_1},\ldots,\norm{H_m}\}+1,
  \quad\text{and}\quad
  n = n_1\cdots n_m
\end{equation}
for $n_k = \dim(\V_k)$.
The following lemma establishes that $G$ can be efficiently approximated.

\begin{lemma}
  There exists a deterministic, polynomial-time algorithm that, given input
  $H_1,\ldots,H_m$, $\rho_1,\ldots,\rho_m$, and $\delta > 0$, outputs
  $(\xi_1,\ldots,\xi_m)\in\C_1\times\cdots\times\C_m$ satisfying
  \begin{equation}
    \bignorm{
      G(\rho_1,\ldots,\rho_m) - (\xi_1,\ldots,\xi_m)}_2 < \delta.
  \end{equation}
\end{lemma}

\begin{proof}
  Let us begin with the approximation of the projections
  $\sigma_k = \op{proj}(\rho_k\,|\,\C_k)$ for each $k\in\{1,\ldots,m\}$.
  For any Hermitian operator $H$, the block operator
  \begin{equation}
    \begin{pmatrix}
      Z & H\\
      H & \I
    \end{pmatrix}
  \end{equation}
  is positive semidefinite if and only if $Z\geq H^2$,
  by the Schur complement theorem.
  Minimizing the trace over all such $Z$ therefore yields
  $\tr(Z)=\norm{H}_2^2$.
  The projection $\sigma_k$ is therefore given by the optimal solution to the
  following semidefinite program:
  \begin{equation}
    \begin{aligned}
      \text{minimize}: \hspace{2mm} & \tr(Z_k)\\
      \text{subject to}: \hspace{2mm} &
      \begin{pmatrix}
        Z_k & \rho_k - Y_k\\
        \rho_k - Y_k & \I_{\V_k}
      \end{pmatrix}
      \geq 0\\
      & Y_k \in \C_k\\
      & Z_k \in \Pos(\V_k).
    \end{aligned}
  \end{equation}
  Specifically, the unique optimal solution $(Y_k,Z_k)$ to this semidefinite
  program satisfies $Y_k = \sigma_k = \op{proj}(\rho_k\,|\,\C_k)$ and
  $\tr(Z_k) = \norm{\rho_k - \sigma_k}_2^2$.

  Through the use of the ellipsoid method, as presented by
  \cite{GroetschelLS88} for instance, one may compute in time polynomial in the
  input length and $\log(1/\eta)$, for any given positive real number $\eta$, a
  feasible solution $(Z_k,Y_k)$ to this semidefinite program that is within
  $\eta$ of its optimal value.
  That is, in polynomial time one may compute $\xi_k\in\C_k$ such that
  \begin{equation}
    \norm{\rho_k - \xi_k}_2^2 \leq \norm{\rho_k - \sigma_k}_2^2 + \eta.
  \end{equation}
  This requires an examination of specific aspects of the semidefinite program
  that are reflected (up to a scalar multiple in the last constraint) by the
  equations \eqref{eq:SDP-strategy-constraints} in Section~\ref{sec:strategies}
  along with a recognition that the feasible region may be bounded.
  The analysis is straightforward and we omit it here.
  
  Now, if it is the case that $\rho_k\in\C_k$, then $\sigma_k = \rho_k$, and we
  conclude immediately that
  \begin{equation}
    \norm{\xi_k - \sigma_k}_2 \leq \sqrt{\eta}.
  \end{equation}
  If $\rho_k\not\in\C_k$, then it follows that
  $\ip{\xi_k - \sigma_k}{\rho_k - \sigma_k} \leq 0$;
  that this inequality holds for every choice of $\xi_k\in\C_k$ is, in fact, a
  well known necessary and sufficient condition for $\sigma_k$ to be the
  projection of $\rho_k$ into $\C_k$.
  By the law of cosines we have
  \begin{equation}
    \norm{\rho_k - \xi_k}_2^2
    = \norm{\rho_k - \sigma_k}_2^2 + \norm{\xi_k - \sigma_k}_2^2
    - 2\ip{\xi_k - \sigma_k}{\rho_k - \sigma_k},
  \end{equation}
  and so we conclude that
  \begin{equation}
    \norm{\rho_k - \xi_k}_2^2
    \geq \norm{\rho_k - \sigma_k}_2^2 + \norm{\xi_k - \sigma_k}_2^2,
  \end{equation}
  which again implies
  \begin{equation}
    \norm{\xi_k - \sigma_k}_2 \leq \sqrt{\eta}.
  \end{equation}
    
  The computation of each $P_k = \op{proj}\bigl(
  \Xi_k(\sigma_{-k}) - \alpha_k \I_{\V_k}\,|\,\K_k)$
  may be performed in almost exactly the same manner, through almost exactly
  the same semidefinite program.
  We note in particular that the optimal value is no larger than
  $\norm{\Xi_k(\sigma_{-k}) - \alpha_k \I_{\V_k}}^2_2$, as the projection
  of $\Xi_k(\sigma_{-k}) - \alpha_k \I_{\V_k}$ onto $\K_k$ is no further away
  from this operator than the zero operator, which is contained in $\K_k$, and
  so once again the feasible region may be bounded.
  Thus we may compute, again in polynomial time, $R_k\in\K_k$
  satisfying $\norm{P_k - R_k}_2 \leq \sqrt{\eta}$.

  All of the other computations required to approximate $G$ can be performed
  exactly.
  The lemma follows by choosing $\eta$ to be sufficiently small while
  polynomial in $\delta$ and the input length to the problem.
\end{proof}

It therefore follows from Corollary~\ref{cor:fixed-point} that, on input
$H_1,\ldots,H_m$ and $\delta>0$, the problem of computing a
$\delta$-approximate fixed point of $G$ is contained in PPAD.
It remains to prove that from such an approximate fixed point of $G$, we
obtain an approximate Nash equilibrium for a game described by
$H_1,\ldots,H_m$.

At this point we face a minor inconvenience: an approximate fixed point
$(\rho_1,\ldots,\rho_m)$ of $G$ provided by the PPAD computation whose
existence is implied by Corollary~\ref{cor:fixed-point} might not be contained
in $\C_1\times\cdots\times\C_m$, although by necessity it will be close.
Because we require an approximate Nash equilibrium to consist of strategies and
not ``near strategies,'' we must project these density operators onto the sets
$\C_1,\ldots,\C_m$.
Specifically, suppose that $(\rho_1,\ldots,\rho_m)$ is an
$(\eta/2)$-approximate fixed point of $G$, for
\begin{equation}
  \eta = \frac{\varepsilon^2}{(3nM)^4}.
\end{equation}
Writing $\sigma_k = \op{proj}(\rho_k\,|\,\C_k)$ for each $k\in\{1,\ldots,m\}$
as before, we find by the definition of $G$ that
\begin{equation}
  G(\rho_1,\ldots,\rho_m) = G(\sigma_1,\ldots,\sigma_m)
  \in \C_1\times\cdots\times\C_m,
\end{equation}
and combining this observation with the fact that projections are 1-Lipschitz,
it follows that $(\sigma_1,\ldots,\sigma_m)$ is also an $(\eta/2)$-approximate
fixed point of $G$.
Although the density operators $(\sigma_1,\ldots,\sigma_m)$ cannot be computed
exactly from $(\rho_1,\ldots,\rho_m)$, the analysis used in the proof of the
previous lemma implies that, in polynomial time, one may compute from 
$(\rho_1,\ldots,\rho_m)$ an $m$-tuple
$(\xi_1,\ldots,\xi_m)\in\C_1\times\cdots\times\C_m$
(with this containment guaranteed by the ellipsoid method)
satisfying
\begin{equation}
  \norm{(\sigma_1,\ldots,\sigma_m) - (\xi_1,\ldots,\xi_m)}_2 <
  \frac{\eta}{4K}.
\end{equation}
It follows that
\begin{equation}
  \begin{multlined}
    \norm{G(\xi_1,\ldots,\xi_m) - (\xi_1,\ldots,\xi_m)}_2\\[2mm]
    \leq
    \norm{G(\xi_1,\ldots,\xi_m) - G(\sigma_1,\ldots,\sigma_m)}_2
    + \norm{G(\sigma_1,\ldots,\sigma_m) - (\sigma_1,\ldots,\sigma_m)}_2\\[1mm]
    + \norm{(\sigma_1,\ldots,\sigma_m) - (\xi_1,\ldots,\xi_m)}_2
    \leq \eta.
  \end{multlined}
\end{equation}
Thus, $(\xi_1,\ldots,\xi_m)$ is an $\eta$-approximate fixed point of $G$.

One more lemma is needed, which will imply that by scaling the density
operators $(\xi_1,\ldots,\xi_m)$, an $\varepsilon$-approximate Nash equilibrium
is obtained.

\begin{lemma}
  \label{lemma:gain-function-fixed-point}
  Let $\C\subseteq\Density(\complex^n)$ be a nonempty, convex, and compact set
  of density operators, let $\K = \op{cone}(\C)$ be the cone generated by $\C$,
  and let $A\in\Herm(\complex^n)$ be a Hermitian operator.
  For a given density operator $\sigma\in\C$, define
  \begin{equation}
    P = \op{proj}(A - \ip{A}{\sigma} \I \,|\, \K),
  \end{equation}
  and assume that
  \begin{equation}
    \label{eq:norm-assumption}
    \biggnorm{
      \frac{\sigma + P}{1 + \tr(P)} - \sigma}_2 \leq \eta
  \end{equation}
  for $\eta> 0$.
  It is the case that
  \begin{equation}
    \ip{A}{\sigma} \geq \sup_{\xi\in\C} \ip{A}{\xi} - \delta
  \end{equation}
  for
  \begin{equation}
    \delta = (1 + 3n\tsp\norm{A})^2 \sqrt{\eta}.
  \end{equation}
\end{lemma}

\begin{proof}
  The operator $P$ is defined to be the closest element of the cone $\K$ to the
  operator $A - \ip{A}{\sigma} \I$ with respect to the Frobenius norm, which is
  to say that
  \begin{equation}
    \label{eq:original-inequality}
    \bignorm{\bigl(A - \ip{A}{\sigma} \I\bigr) - P}_2
    \leq \bignorm{\bigl(A - \ip{A}{\sigma} \I\bigr) - \lambda\xi}_2
  \end{equation}
  for every choice of $\lambda\geq 0$ and $\xi\in\C$.
  We may first consider the case that $\lambda = 0$, from which the bound
  \begin{equation}
    \norm{P}_2
    \leq 2\, \bignorm{A - \ip{A}{\sigma}\I}_2
    \leq 4 \sqrt{n} \,\norm{A},
  \end{equation}
  is obtained, implying that $\tr(P) \leq 4n\norm{A}$.
  It follows that
  \begin{equation}
    \norm{P - \tr(P)\sigma}_2
    = (1 + \tr(P)) \biggnorm{ \frac{\sigma + P}{1 + \tr(P)} - \sigma}_2
    \leq \bigl(1 + 4 n \norm{A}\bigr) \eta.
  \end{equation}

  Next, by squaring both sides of the inequality \eqref{eq:original-inequality}
  and simplifying, one obtains
  \begin{equation}
    \label{eq:inequality1}
    \lambda \bigip{A - \ip{A}{\sigma} \I}{\xi}
    \leq \bigip{A - \ip{A}{\sigma} \I}{P}
    + \frac{\lambda^2}{2} \norm{\xi}_2^2 - \frac{1}{2}\norm{P}_2^2.
  \end{equation}
  Disregarding the negative final term and observing the inequality
  $\norm{\xi}_2\leq 1$ and the equality $\bigip{A-\ip{A}{\sigma}\I}{\sigma}=0$,
  we find that
  \begin{equation}
    \label{eq:inequality2}
    \lambda \bigip{A - \ip{A}{\sigma} \I}{\xi}
    \leq \bigip{A-\ip{A}{\sigma}\I}{P-\tr(P)\sigma} + \frac{\lambda^2}{2},
  \end{equation}
  again for every $\lambda\geq 0$ and $\xi\in\C$.
  Setting $\lambda = \sqrt{\eta}$ and applying the Cauchy--Schwarz
  inequality yields
  \begin{equation}
    \begin{aligned}
      \ip{A}{\xi} - \ip{A}{\sigma}
      & = \bigip{A - \ip{A}{\sigma} \I}{\xi}\\
      & \leq \frac{1}{\sqrt{\eta}}
      \bignorm{A - \ip{A}{\sigma}\I}_2
      \bignorm{P - \tr(P)\sigma}_2 + \frac{\sqrt{\eta}}{2}\\
      & \leq 
      \Bigl(2 \sqrt{n}\tsp\norm{A} (1 + 4n\tsp\norm{A}) + \frac{1}{2}\Bigr)
      \sqrt{\eta}\\
      & \leq (1 + 3n\tsp\norm{A})^2 \sqrt{\eta}.
    \end{aligned}
  \end{equation}
  As this bound holds for every $\xi\in\C$, the lemma is proved.
\end{proof}

We conclude from this lemma that
\begin{equation}
  \label{eq:sup-bound-1}
  \bigip{\Xi_k(\xi_{-k})}{\xi_k}
  \geq
  \sup_{\tau\in\C_k} \bigip{\Xi_k(\xi_{-k})}{\tau} -
  \bigl(1 + 3 n_k \norm{\Xi_k(\xi_{-k})}\bigr)^2 \sqrt{\eta}.
\end{equation}
Define $Q_k = d_k \xi_k$ for each $k\in\{1,\ldots,m\}$ so that
\begin{equation}
  (Q_1,\ldots,Q_m) \in \S_1\times\cdots\times\S_m.
\end{equation}
By \eqref{eq:sup-bound-1} it follows that
\begin{equation}
  \begin{aligned}
    \bigip{\Xi_k(Q_{-k})}{Q_k}
    & \geq \sup_{R\in\S_k} \bigip{\Xi_k(Q_{-k})}{R} -
    d_1\cdots d_m \bigl(1 + 3 n_k \norm{\Xi_k(\sigma_{-k})}\bigr)^2
    \sqrt{\eta}\\
    & \geq \sup_{R\in\S_k} \bigip{\Xi_k(Q_{-k})}{R} - \varepsilon,
  \end{aligned}    
\end{equation}
and therefore $(Q_1,\ldots,Q_m)$ is an $\varepsilon$-approximate Nash
equilibrium of the interactive quantum game having associated payoff operators
$H_1,\ldots,H_m$.
As $(Q_1,\ldots,Q_m)$ has been obtained from $\varepsilon$ together
with the approximate fixed point $(\rho_1,\ldots,\rho_m)$ of $G$ by a
polynomial-time computation,
Theorem~\ref{theorem:approximate-quantum-Nash-in-PPAD} is proved.


\section{Discussion of directions for further research}

We conclude the paper with a collection of open problems and suggestions
of topics that we hope might inspire further work on quantum game theory and
its connections to theoretical computer science.

\begin{mylist}{8mm}
\item[1.]
  Is there a quantum extension or variant of the Lemke--Howson algorithm
  \cite{LemkeH1964} for computing or approximating a Nash equilibrium in a
  non-interactive two-player quantum game?

\item[2.]
  It is interesting to consider quantum players having different restrictions
  placed on their strategies.
  For example, we might insist that players process quantum information using
  limited resources, or restrict player's actions so that they represent
  adversarial models of noise.
  Along similar lines, one may consider alternative ways of describing the
  referee's actions, such as by quantum circuits.  
  What can be said about quantum games in contexts such as these?

\item[3.]
  We have limited our focus to a non-cooperative setting, in which players must
  play independently, representing an inability for the players to form
  collusions.
  The consideration of collusions, and more generally the study of
  \emph{cooperative quantum game theory}, is an interesting research
  direction.

  For instance, let us imagine that there exists a shared quantum state that
  allows players to implement a strategy in a quantum game that is good by some
  measure.
  Nonlocal games, for instance, may naturally be viewed as non-interactive
  games in a purely cooperative setting where shared quantum states can lead
  to improved strategies.
  In the general, not completely cooperative setting, such a shared state
  could be provided by a trusted non-participant in the game, like in the work
  of Zhang~\cite{Zhang2012} on correlated (or entangled) equilibria.
  An alternative is a setting in which such a state must arise from an
  unmediated interaction between colluding players, in which case players could
  deviate from any prescribed protocol that produces this state.
  
\item[4.]
  Closely related to the notion of an unmediated interaction, one may consider
  games in which there is no referee.
  Coin-flipping may be cast as an example, and its evidently complicated
  structure suggests nothing less in a setting in which the goal is,
  perhaps, to produce a quantum state of interest.

\item[5.]
  Generally speaking, can quantum game theory provide a foundation through
  which one may discover quantum protocols having either theoretical or
  practical utility?
  
\end{mylist}


\subsection*{Acknowledgments}

This research was undertaken thanks in part to funding from the Canada First
Research Excellence Fund.
We thank Sanketh Menda for helpful suggestions at an early stage of this work.


\bibliographystyle{quantum}

\begin{thebibliography}{10}

\bibitem{Meyer1999}
David Meyer.
\newblock ``Quantum strategies''.
\newblock \href{https://dx.doi.org/10.1103/PhysRevLett.82.1052}{Physical Review
  Letters {\bf 82}, 1052--1055}~(1999).

\bibitem{EisertWL1999}
Jens Eisert, Martin Wilkens, and Maciej Lewenstein.
\newblock ``Quantum games and quantum strategies''.
\newblock \href{https://dx.doi.org/10.1103/PhysRevLett.83.3077}{Physical Review
  Letters {\bf 83}, 3077--3080}~(1999).

\bibitem{vonNeumannM1953}
John von Neumann and Oskar Morgenstern.
\newblock ``Theory of games and economic behavior''.
\newblock \href{https://dx.doi.org/10.1515/9781400829460}{Princeton University
  Press}. ~(1953).
\newblock third edition.

\bibitem{Nash1950}
John Nash.
\newblock ``Equilibrium points in \emph{n}-person games''.
\newblock \href{https://dx.doi.org/10.1073/pnas.36.1.48}{Proceedings of the
  National Academy of Sciences {\bf 36}, 48--49}~(1950).

\bibitem{Nash1950-thesis}
John Nash.
\newblock ``Non-cooperative games''.
\newblock PhD thesis.
\newblock Princeton University.
\newblock ~(1950).

\bibitem{GuoZK2008}
Hong Guo, Juheng Zhang, and Gary Koehler.
\newblock ``A survey of quantum games''.
\newblock \href{https://dx.doi.org/10.1016/j.dss.2008.07.001}{Decision Support
  Systems {\bf 46}, 318--332}~(2008).

\bibitem{vanEnkP2002}
Steven van Enk and Rob Pike.
\newblock ``Classical rules in quantum games''.
\newblock \href{https://dx.doi.org/10.1103/PhysRevA.66.024306}{Physical Review
  A {\bf 66}, 024306}~(2002).

\bibitem{Wu2004a}
Jinshan Wu.
\newblock ``A new mathematical representation of game theory {I}''.
\newblock Unpublished manuscript, arXiv:quant-ph/0404159~(2004).

\bibitem{Wu2004b}
Jinshan Wu.
\newblock ``A new mathematical representation of game theory {II}''.
\newblock Unpublished manuscript, arXiv:quant-ph/0405183~(2004).

\bibitem{Zhang2012}
Shengyu Zhang.
\newblock ``Quantum strategic game theory''.
\newblock In Proceedings of the 3rd Innovations in Theoretical Computer Science
  Conference.
\newblock \href{https://dx.doi.org/10.1145/2090236.2090241}{Pages 39--59}.
\newblock ~(2012).

\bibitem{GutoskiW2007}
Gus Gutoski and John Watrous.
\newblock ``Toward a general theory of quantum games''.
\newblock In Proceedings of the 39th Annual ACM Symposium on Theory of
  Computing.
\newblock \href{https://dx.doi.org/10.1145/1250790.1250873}{Pages 565--574}.
\newblock ~(2007).

\bibitem{ChiribellaDP2008}
Giulio Chiribella, Giacomo D'Ariano, and Paolo Perinotti.
\newblock ``Quantum circuit architecture''.
\newblock \href{https://dx.doi.org/10.1103/PhysRevLett.101.060401}{Physical
  Review Letters {\bf 101}, 060401}~(2008).

\bibitem{ChiribellaDP2009}
Giulio Chiribella, Giacomo D'Ariano, and Paolo Perinotti.
\newblock ``Theoretical framework for quantum networks''.
\newblock \href{https://dx.doi.org/10.1103/PhysRevA.80.022339}{Physical Review
  A {\bf 80}, 022339}~(2009).

\bibitem{DaskalakisGP2009}
Constantinos Daskalakis, Paul Goldberg, and Christos Papadimitriou.
\newblock ``The complexity of computing a {Nash} equilibrium''.
\newblock \href{https://dx.doi.org/10.1137/070699652}{SIAM Journal on Computing
  {\bf 39}, 195--259}~(2009).

\bibitem{ChenDT2009}
Xi~Chen, Xiaotie Deng, and Shang-Hua Teng.
\newblock ``Settling the complexity of computing two-player {Nash}
  equilibria''.
\newblock \href{https://dx.doi.org/10.1145/1516512.1516516}{Journal of the ACM
  {\bf 56}, 14}~(2009).

\bibitem{Papadimitriou1994}
Christos Papadimitriou.
\newblock ``On the complexity of the parity argument and other inefficient
  proofs of existence''.
\newblock \href{https://dx.doi.org/10.1016/S0022-0000(05)80063-7}{Journal of
  Computer and system Sciences {\bf 48}, 498--532}~(1994).

\bibitem{EtessamiY2010}
Kousha Etessami and Mihalis Yannakakis.
\newblock ``On the complexity of {Nash} equilibria and other fixed points''.
\newblock \href{https://dx.doi.org/10.1137/080720826}{SIAM Journal on Computing
  {\bf 39}, 2531--2597}~(2010).

\bibitem{GibbonsHW2004}
Kathleen Gibbons, Matthew Hoffman, and William Wootters.
\newblock ``Discrete phase space based on finite fields''.
\newblock \href{https://dx.doi.org/10.1103/PhysRevA.70.062101}{Physical Review
  A {\bf 70}, 062101}~(2004).

\bibitem{Gross2006}
David Gross.
\newblock ``Hudson's theorem for finite-dimensional quantum systems''.
\newblock \href{https://dx.doi.org/10.1063/1.2393152}{Journal of Mathematical
  Physics {\bf 47}, 122107}~(2006).

\bibitem{AroraB2009}
Sanjeev Arora and Boaz Barak.
\newblock ``Computational complexity: A modern approach''.
\newblock \href{https://dx.doi.org/10.1017/CBO9780511804090}{Cambridge
  University Press}. ~(2009).

\bibitem{NielsenC2000}
Michael Nielsen and Isaac Chuang.
\newblock ``Quantum computation and quantum information''.
\newblock \href{https://dx.doi.org/10.1017/CBO9780511976667}{Cambridge
  University Press}. ~(2000).

\bibitem{Wilde2017}
Mark Wilde.
\newblock ``Quantum information theory''.
\newblock \href{https://dx.doi.org/10.1017/9781316809976}{Cambridge University
  Press}. ~(2017).
\newblock second edition.

\bibitem{Watrous2018}
John Watrous.
\newblock ``Theory of quantum information''.
\newblock \href{https://dx.doi.org/10.1017/9781316848142}{Cambridge University
  Press}. ~(2018).

\bibitem{Bredon1993}
Glen Bredon.
\newblock ``Topology and geometry''.
\newblock \href{https://dx.doi.org/10.1007/978-1-4757-6848-0}{Volume 139 of
  Graduate Texts in Mathematics}.
\newblock Springer. ~(1993).

\bibitem{Glicksberg1952}
Irving Glicksberg.
\newblock ``A further generalization of the {Kakutani} fixed point theorem,
  with application to {Nash} equilibrium points''.
\newblock \href{https://dx.doi.org/10.2307/2032478}{Proceedings of the American
  Mathematical Society {\bf 3}, 170--174}~(1952).

\bibitem{Nash1951}
John Nash.
\newblock ``Non-cooperative games''.
\newblock \href{https://dx.doi.org/10.2307/1969529}{Annals of Mathematics,
  Second Series {\bf 54}, 286--295}~(1951).

\bibitem{GroetschelLS88}
Martin Gr\"otschel, Laszlo Lov\'asz, and Alexander Schrijver.
\newblock ``Geometric algorithms and combinatorial optimization''.
\newblock
  \href{https://dx.doi.org/10.1007/978-3-642-78240-4}{Springer--Verlag}.
  ~(1988).

\bibitem{LemkeH1964}
Carlton Lemke and Joseph Howson.
\newblock ``Equilibrium points of bimatrix games''.
\newblock \href{https://dx.doi.org/10.1137/0112033}{Journal of the Society for
  industrial and Applied Mathematics {\bf 12}, 413--423}~(1964).

\end{thebibliography}

\end{document}